%% file: patchwork.tex
\documentclass{article} 
\usepackage{preprint,times}

\input{math_commands.tex}

\usepackage{hyperref}
\usepackage{url}
\usepackage{graphicx}
\usepackage{booktabs}
\usepackage{amsmath, amsthm}
\usepackage{subcaption}
\usepackage{pifont}
\usepackage{makecell}
\usepackage{listings}
\usepackage{wrapfig}
\usepackage{orcidlink}

\theoremstyle{plain}
\newtheorem{theorem}{Theorem}
\newtheorem{proposition}{Proposition}

\newtheorem{definition}{Definition}
\usepackage{xcolor}
\usepackage{multirow}

\definecolor{codegreen}{rgb}{0,0.6,0}
\definecolor{codegray}{rgb}{0.5,0.5,0.5}
\definecolor{codepurple}{rgb}{0.58,0,0.82}
\definecolor{backcolour}{rgb}{0.95,0.95,0.92}

\title{Patchwork: A Polyhedral Implicit Representation for Highly Detailed 3D Shapes}

\author{Ruichen Zheng \orcidlink{0009-0008-3657-351X} , Biao Zhang \orcidlink{0000-0001-5685-6092}, Michael Birsak \orcidlink{0000-0001-6375-8124}, Mikhail Skopenkov \orcidlink{0000-0003-2453-0009} \& Peter Wonka \orcidlink{0000-0003-0627-9746} \\
King Abdullah University of Science and Technology \\
Thuwal, 23955, Saudi Arabia \\
\texttt{zheng$\cdot$ruichen@kaust$\cdot$edu$\cdot$sa}, \texttt{biao$\cdot$z@outlook$\cdot$com}, \\
\texttt{michael$\cdot$birsak@kaust$\cdot$edu$\cdot$sa},
\texttt{\{mikhail$\cdot$skopenkov,pwonka\}@gmail$\cdot$com}, 
}

\usepackage{ulem}

\newcommand{\journalvsarxiv}[2]{#2}

\input{fig/figure_macros}

\iclrfinalcopy 
\begin{document}

\maketitle
\fancyhead[L]{Preprint}

\input{sec/0_abstract}
\input{sec/1_intro}
\input{sec/2_related_work}
\input{sec/4_method_optimization}
\input{sec/6_results}
\input{sec/5_conclusions}

\subsection*{AI use statement}
In this work, we have not used generative AI tools to help develop theoretical models or conceptual frameworks, formulate mathematical claims, provide critical ingredients for proving mathematical claims, assist in the writing of proofs, propose or refine hypotheses, design or provide feedback on research methodology or experiments, implement methods, or interpret results. Generating synthetic data sets, assisting with translation, cleaning and reformatting datasets, and supporting qualitative and thematic data analysis are not applicable to this work. We used generative AI tools to help with writing and reviewing the manuscript for consistency and grammar mistakes. We have reviewed, fact-checked, and revised all AI-assisted texts. We take responsibility for the final content of this work, including text, claims, or artifacts produced with the aid of generative AI.

\bibliography{patchwork}
\bibliographystyle{preprint}

\appendix

\input{sec/6_appendix.tex}

\end{document}

%% file: math_commands.tex
\usepackage{amsmath,amsfonts,bm}

\def\eqref#1{equation~\ref{#1}}

\def\1{\bm{1}}

\DeclareMathAlphabet{\mathsfit}{\encodingdefault}{\sfdefault}{m}{sl}
\SetMathAlphabet{\mathsfit}{bold}{\encodingdefault}{\sfdefault}{bx}{n}

\DeclareMathOperator*{\argmax}{arg\,max}

%% file: fig/figure_macros.tex
\newlength{\cmpimgw}
\newlength{\cmpspace}
\newlength{\cmplabelw}


%% file: sec/0_abstract.tex
\begin{abstract}
We introduce patchwork, a highly expressive shape representation that explicitly encodes 3D shapes as unions of polyhedra while admitting a simple closed-form zero-level-set equation. It is grounded in a rigorous mathematical framework that yields provable complexity bounds and guarantees approximation of arbitrary closed surfaces to arbitrary precision. To fit patchwork to a given shape, we pair conventional gradient-based optimization with a novel combinatorial matching scheme. Experiments show that our approach achieves superior quality with 3--6$\times$ fewer parameters compared to existing alternatives on average, and scales to highly detailed artist-created assets. Patchwork further supports explicit polygonalization and conforming volume mesh extraction, making it a versatile, scalable, yet compact representation for downstream tasks, with potential for geometric learning. Project page: \url{https://ankbzpx.github.io/patchwork-page/}.
\end{abstract}

%% file: sec/1_intro.tex
\begin{figure}[h]
    \vspace{-15pt}
    \centering
    \includegraphics[width=0.9\linewidth]{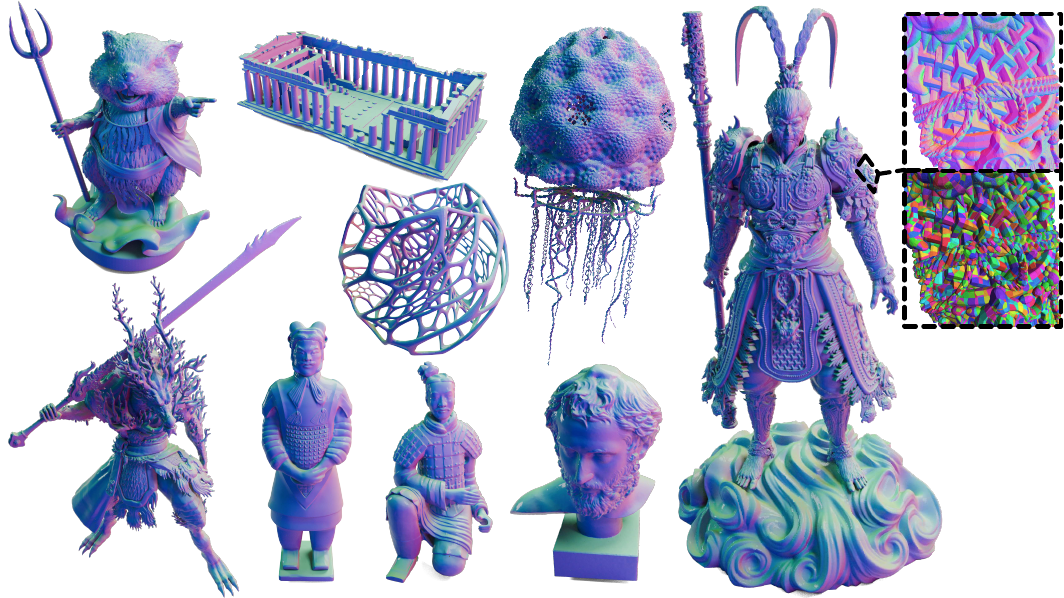}
    \vspace{-5pt}
    \caption{A collection of high-quality shapes~\citep{Zhang_2025_ICCV} expressed in our representation. The left 8 cases (explicit polygonalization) average under 4MB of parameters, while the right (Marching Cubes at 1536³ resolution) uses 40MB to represent the mesh over 800 MB. Insets show normals and polyhedra (randomly colored) visualized via ray marching over our implicit form.}
    \label{fig:teaser}
    \vspace{-10pt}
\end{figure}

\section{Introduction}
\label{sec:intro}

Recent advances in modern graphics center around the revolution in 3D representations. Neural implicit representations encode shapes as continuous functions~\citep{park2019deepsdf}, circumventing the cubic memory complexity of discrete voxel grids and achieving unrivaled representational expressiveness~\citep{sitzmann2020implicit}; volume density fields~\citep{mildenhall2020nerfrepresentingscenesneural} and 3D Gaussian representations~\citep{kerbl20233d} enable novel-view synthesis from images, bypassing the conventional modeling–texturing–rendering pipeline. Recently, foam-based 3D representations have attracted considerable attention for their competitive ray-tracing performance~\citep{govindarajan2025radiantfoamrealtimedifferentiable, govindarajan2026powerfoamunifyingrealtime}. Under such a representation, the space is partitioned into convex polyhedral cells, typically via Voronoi or Power tessellation, with each cell associated with additional attributes such as sign or volume density, and the surface geometry is then encoded as a subset of the cell interfaces. Although loosely characterized as a hybrid implicit–explicit representation~\citep{rech2025sdfoamsigneddistancefoamexplicit}, it lacks a precise implicit formulation and a theoretical analysis of its expressive power. Moreover, despite its geometric learning and meshing potential~\citep{maruani2023voromesh}, fitting a foam-based representation to an existing geometry remains challenging: current fitting methods yield coarse, low-resolution results~\citep{maruani2023voromesh, charawi20263dv}, with demonstrated expressiveness significantly outclassed by neural implicit variants.

In this work, we revisit \textit{Viro's patchworking} idea from pure mathematics~\citep{IV-96} to derive a formal yet simple expression for the implicit foam-based surface representation, which we term \textit{patchwork} in acknowledgment of its origin. In patchworks, we unify the convex cells as power cells, encode them algebraically as linear functions, each signed and composed via a global maximum, resulting in a closed-form zero level set equation. We derive its complexity bounds and prove its ability to approximate arbitrary closed surfaces with arbitrary precision. 

To realize patchworks's expressiveness in practice, we develop a novel gradient-based optimization scheme: In contrast to the conventional practice of replacing the maximum with a soft variant (e.g., LogSumExp~\citep{Deng-etal-20}) to distribute gradient signals over multiple cells, we keep the hard maximum and match the best-fit cells with respect to input constraints to ensure clean gradient updates. Our scheme is simple to implement, retains high fitting throughput, and is compatible with deep learning optimizers such as Adam. We evaluate our method on multiple datasets of varying shape complexity, showing that patchwork has ample capacity to encode detailed geometry, surpassing highly expressive neural implicit counterparts such as SIREN~\citep{sitzmann2020implicit}, while requiring 3--6$\times$ less raw floating-point storage on disk on average. We further demonstrate patchwork's scaling capacity across a set of complex structures, as well as practical applications unique to our representation.

Our contributions are summarized as follows:
\begin{itemize}
    \item We formalize the implicit foam-based surface representation in closed form. Our representation, patchwork, bridges neural networks and known mathematical theories, and is in theory 
    capable of approximating arbitrary closed surfaces with arbitrary precision.
    \item We propose a novel optimization scheme that makes fitting our representation viable with gradient descent, yielding results superior to baselines while requiring, on average, 3–6$\times$ fewer floating-point parameters.
\end{itemize}

%% file: sec/2_related_work.tex
\section{Related work}
\label{sec:related-work}
In this section, we focus on reviewing the literature from the machine learning side and defer the mathematical origins to Appendix~\ref{sec:origin}.

\textbf{Neural shape representations.}
The representation of 3D shapes using neural networks has been explored through various approaches~\citep{neuralfieldsurvey}. Early methods focus on designing losses to regularize the distance field, notably the sign-agnostic variants~\citep{atzmon2020salsignagnosticlearning, atzmon2020saldsignagnosticlearning, lipman2021phasetransitionsdistancefunctions} and the Eikonal constraint~\citep{gropp2020implicitgeometricregularizationlearning}. Others improve the network's representational power with architectural modifications, such as the sine activation function~\citep{sitzmann2020implicit} and positional encoding~\citep{mildenhall2020nerfrepresentingscenesneural, tancik2020fourier}. Another line of research relocates the expressiveness of neural networks into auxiliary parameter-heavy spatial data structures. NGLOD~\citep{takikawa2021neural} organizes features in a sparse octree, providing adaptive levels of detail, whereas InstantNGP~\citep{muller2022instant} uses a multiresolution hash table and achieves extremely fast training speed. More relevant to ours are methods that represent shape as a composition of primitives. \cite{tulsiani2017learning} assembles cuboids, CvxNet~\citep{Deng-etal-20} composes shapes as convex cells, and BSP-Net~\citep{chen2020bsp} performs recursive binary space partitioning. Although we also use convex cells, our decomposition is conforming, and we partition the full space using power diagrams~\citep{aurenhammer1987power} without recursion.

\textbf{Foam-based fitting methods.} VoroMesh~\citep{maruani2023voromesh} pioneers the gradient-based fitting of explicit Voronoi surface meshes, but its parameter count scales cubically with grid resolution and its reconstructed surfaces exhibit large seams. DCCVT~\citep{charawi20263dv} initializes from a pre-fitted implicit signed distance function (SDF) and clips the Voronoi cells against the jointly optimized zero-level set, reducing those artifacts. RadiantFoam~\citep{govindarajan2025radiantfoamrealtimedifferentiable} tessellates the space with Voronoi cells and defines per-cell density for volume rendering. PowerFoam~\citep{govindarajan2026powerfoamunifyingrealtime} generalizes it with power cells and bisects each cell to capture macro-scale geometry. Neither of them supports explicit polygonalization and requires Truncated-SDF fusion~\citep{kinectfusion} for surface extraction. SDFoam~\citep{rech2025sdfoamsigneddistancefoamexplicit} fits a guiding implicit SDF along Voronoi sites and extracts the surface as the boundary of the union of cells whose centers have negative signed distance. In contrast, our method yields a power diagram that decomposes the full space, without clipping or relying on a guiding field, and supports explicit polygonalization. Tab.~\ref{tab:feature-table} summarizes the full feature comparison. To the best of our knowledge, PowerCrust~\citep{powercrust} is the only classical method that, like ours, extracts the surface directly from the faces of a power diagram (rather than dual triangulation, restricted or bounded variants), but it does not offer an implicit representation and requires approximating the medial axis, a computation that is notoriously difficult to make robust and scale up~\citep{wang2024mattopotopologypreservingmedialaxis, Wang_2025}.

\begin{figure}[t]
\raisebox{10pt}{
  \begin{minipage}[b]{0.65\linewidth}
    \centering
    \captionof{table}{Feature table of foam-based methods.}
    \label{tab:feature-table}
    \resizebox{\linewidth}{!}{
    \begin{tabular}{lcccc}
        & \bf \makecell{Clipping-\\free} & \bf \makecell{No aux.~\\field} & \bf \makecell{Explicit~\\polygons} & \bf Scalable \\
        \hline \\
        VoroMesh~\citep{maruani2023voromesh}  & \ding{51} & \ding{51} & \ding{51} & \ding{55} \\
        DCCVT~\citep{charawi20263dv} & \ding{55} & \ding{55} & \ding{51} & \ding{55} \\
        
        RadiantFoam~\citep{govindarajan2025radiantfoamrealtimedifferentiable} & \ding{51} & \ding{51} & \ding{55} & \ding{51} \\
        PowerFoam~\citep{govindarajan2026powerfoamunifyingrealtime}
        & \ding{51} & \ding{51} & \ding{55} & \ding{51} \\
        SDFoam~\citep{rech2025sdfoamsigneddistancefoamexplicit} & \ding{51} & \ding{55} & \ding{51} & \ding{51} \\
      Ours   & \ding{51} & \ding{51} & \ding{51} & \ding{51}
    \end{tabular}}
  \end{minipage}}
  \hfill
  \raisebox{10pt}{
  \begin{minipage}[t]{0.3\linewidth}
    \centering
    \includegraphics[width=0.9\linewidth]{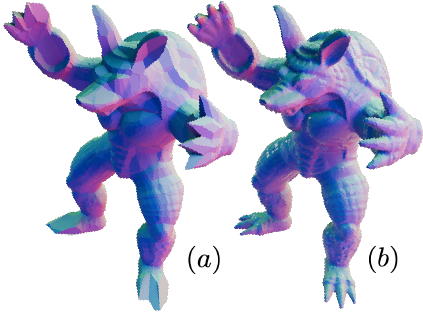}
    \vspace{-5pt}
    \captionof{figure}{Fitting result (a) without and (b) with our discrete matching.}
    \label{fig:ablation}
  \end{minipage}}
  \vspace{-15pt}
\end{figure}

%% file: sec/4_method_optimization.tex
\begin{figure}[h]
    \centering
    \includegraphics[width=0.95\linewidth]{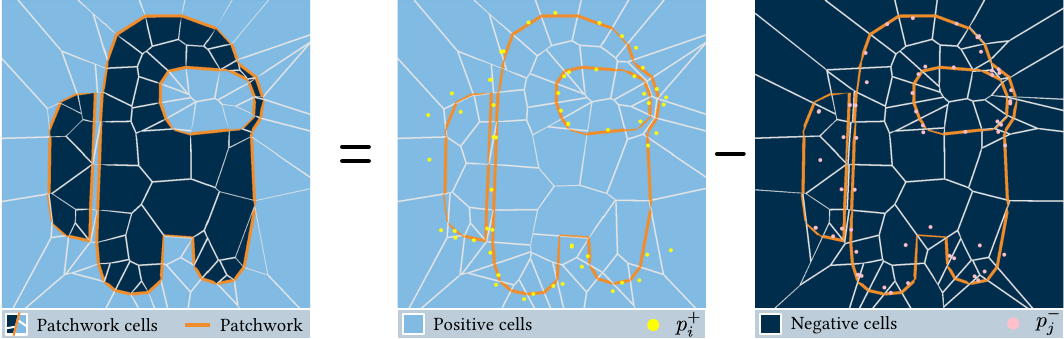}
    \caption{Illustration of a 2D patchwork: subtracting the piecewise linear functions associated with the dark (negative) power diagram from the light (positive) one yields the patchwork (orange curve). The yellow/pink points are the sites $p^+_i$ / $p^-_j$ of the positive/negative power diagrams, respectively.
    }
    \label{fig:illustrative}
    \vspace{-15pt}
\end{figure}

\newcommand{\kdp}{\mathbf{\dot {\mathit k}}^+}
\newcommand{\kdm}{\mathbf{\dot {\mathit k}}^-}
\newcommand{\kdpm}{\mathbf{\dot {\mathit k}}^\pm}
\newcommand{\kddp}{\mathbf{\ddot {\mathit k}}^+}
\newcommand{\kddm}{\mathbf{\ddot {\mathit k}}^-}
\newcommand{\kddpm}{\mathbf{\ddot {\mathit k}}^\pm}

\section{Method}
\label{sec:mainmethod}
In this section, we introduce our fitting method.
Given a target surface in a representation that lets us sample point-normal pairs $(x_k, n_k)$, we aim to fit a field $f(x)$ whose zero-level set describes the surface, so that $f(x_k) = 0$ and $\nabla f(x_k) \parallel n_k$. We employ fields of the following special form.

\begin{definition} 
\label{def-patchwork-general} 
(See Fig.~\ref{fig:simple}.)
A \emph{patchwork field} is a function $\mathbb{R}^3\to\mathbb{R}$ of the form

\vspace{-10pt}
\begin{equation}
\label{eq:patchwork-defintion} 
f(x)= \max_{1 \leq i \leq N^+} \{\underbrace{\langle a^+_i, x \rangle + c^+_i}_{l_i^+(x)}\} - \max_{1 \leq j \leq N^-} \{\underbrace{\langle a^-_j, x \rangle + c^-_j}_{l_j^-(x)}\},
\end{equation}
\vspace{-10pt}

with $N^+, N^-\!\in\!\mathbb{N}$, parameters $a_i^+, a_j^- \!\in\! \mathbb{R}^3$, $c_i^+, c_j^-\!\in\! \mathbb{R}$. A \emph{patchwork} is the zero level set $f(x)\!=\!0$.
\end{definition}
For each $i=1,\dots,N^+$, the locus of points $x\in\mathbb{R}^3$ where the first maximum in Eq.~(\ref{eq:patchwork-defintion}) is attained at this particular~$i$, i.e. $l_i^+(x)=\max_{k}l^+_k(x)$, is called the $i$-th \emph{positive cell} (Fig.~\ref{fig:illustrative} middle). Analogously, we define $j$-th \emph{negative cell} (Fig.~\ref{fig:illustrative} right).
The patchwork field is a linear function in the intersection of a positive and a negative cell, hence is a piecewise linear function.
We also define a \emph{patchwork cell} as the set of all points $x$ where a selected linear function $l_i^+(x)$ or $l_j^-(x)$ is greater or equal than all the other ones, both $l_k^+(x)$ and $l_k^-(x)$ (Fig.~\ref{fig:illustrative} left).
A step-by-step introduction of these core concepts, theoretical expressiveness, and complexity analysis are collected in Appendices~\ref{sec:method}
--\ref{sec:representation-complexity}. (Even when essentially known, those results are difficult to extract from existing literature).

\begin{figure}[h]
    \centering
    \includegraphics[width=0.25\linewidth]{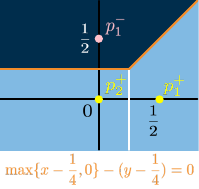}
    \vspace{-10pt}
    \caption{A 
    patchwork (orange) and the sites (yellow/pink) and cells (blue) of the
    power diagram}
    \label{fig:simple}
\end{figure}

\textbf{Geometric Interpretation.} 
Eq.~(\ref{eq:patchwork-defintion}) can be rewritten by completing the squares as follows:
\begin{equation}
    \label{eq:power-diagram}
    f(x) = \min_{1 \leq j \leq N^-} \{\underbrace{|x - p_j^-|^2 - w_i^-}_{|x|^2-l_j^-(x)}\} - \min_{1 \leq i \leq N^+} \{\underbrace{|x - p_i^+|^2 - w_j^+}_{|x|^2-l_i^+(x)}\},
\end{equation}
where $p^+_i = \frac{1}{2} a^+_i, p^-_j = \frac{1}{2} a^-_j \in \mathbb{R}^3$ and $w^+_i = \frac{1}{4} |a^+_i|^2 + c^+_i, w^-_j= \frac{1}{4} |a^-_j|^2 + c^-_j \in \mathbb{R}$ are called \emph{sites} and \emph{weights}, respectively.
Each term in the first minimum is known as the \emph{power distance} between the points $x$ and $p^-_j$. It becomes the usual squared distance when $w^-_j=0$. The locus of points~$x$ closest to a site $p^-_j$ with respect to the power distance is a \emph{cell of the power diagram}. It becomes a Voronoi diagram when $w^-_j=0$ (Fig.~\ref{fig:simple}). 
Thus, our positive/negative/patchwork cells are the cells of three power diagrams (Fig.~\ref{fig:illustrative}) with sites $\{p_i^+\}_{i=1}^{N^+}$, $\{p_j^-\}_{j=1}^{N^-}$, and $\{p_i^+\}_{i=1}^{N^+}\cup\{p_j^-\}_{j=1}^{N^-}$, respectively.

\textbf{Optimization Challenges.}
A standard approach to similar fitting problems is to supervise or regularize with off-surface samples. However, patchwork is not a distance field, making pre-acquiring ground-truth values or applying the Eikonal regularization~\citep{gropp2020implicitgeometricregularizationlearning} infeasible. Next, its piecewise nature allows field values at two query points $x_{j}, x_{k}$ to come from the same pair of linear functions (when $x_{j}, x_{k}$ lie in the same positive and negative cells). These query points, therefore, share the same gradient, 
leaving no capacity to model two possibly conflicting directional constraints $n_{j}, n_{k}$.
(Fig.~\ref{fig:discrete-matching}~(a)).
Replacing the hard maximum with a soft variant would further blend incompatible constraints into neighboring cells and degrade the results (Fig.~\ref{fig:ablation}~(a)).
To this end, departing from the conventional practice of making everything smooth, we deliberately keep the maximum assignment \textit{as-discrete-as-possible}.

\textbf{Algorithm Overview.} 
We initialize the patchwork field using the power diagram (\textbf{Geometric Initialization}).
Then we iteratively refine it using the following five steps.
First, we sample a set of $K$ point-normal pairs $(x_k,n_k)$ viewed as constraints.
Second, we select a matching pair of cells (one positive, one negative) for each of the constraints $(x_k,n_k)$ from four candidates (\textbf{Discrete Matching}).
Third, we stabilize the assignment to limit oscillations in the discrete matching (\textbf{Assignment Stabilization}).
Fourth, we fit the surface so that all constraints $(x_k,n_k)$ are covered by their matched pair of cells (\textbf{Cell Fitting}).
Fifth, we identify and remove redundant linear functions (\textbf{Pruning}). This step includes \textbf{Normal Inversion Handling}. Algorithm pseudocode
is provided in Appendix~\ref{sec:pseudo-code}.

\textbf{Geometric Initialization.}
Inspired by VoroCrust~\citep{Abdelkader_2020}, given a set of samples $\{(x_i, n_i)\}_{i=1}^{N}$, we place the zero-weighted sites to double-cover the surface, that is, set
\begin{equation}
    N^\pm=N,\qquad 
    a_i^\pm = 2 (x_i \pm \delta \cdot n_i), \qquad  c_i^\pm = -\frac{1}{4} |a_i^\pm|^2,
\end{equation}
where $\delta = \frac{0.01}{\tau}$ is a small displacement for a shape normalized to $[-1, 1]^3$, and $\tau \in \mathbb{R}$ is a scaling hyperparameter based on the geometric feature scale. The choice of $\tau$ is discussed in Appendix~\ref{sec:tau-ablation}.

\textbf{Discrete Matching.}
For each query $(x_k, n_k)$, we choose a matching pair of cells such that the difference of corresponding linear functions is best aligned with the normal $n_k$: one from the positive cell containing $x_k$ and its neighbors, and one from the negative cell containing $x_k$ and its neighbors.
Since it is difficult to reconstruct the cell neighborhood structure, we only choose among the closest neighboring cells, i.e., the ``second-maximum'' ones, that are guaranteed to be adjacent if nonempty (Fig.~\ref{fig:discrete-matching}(b)).
Namely, for each fixed $k$, take the values of $i$ for which $l_i^+(x_k)=\langle a_i^+,x_k\rangle +c_i$ is maximal and second maximal. These indices are denoted $\kdp$ and $\kddp$.
Analogously, define $\kdm$ and $\kddm$.
We choose
\begin{equation}
(k^+, k^-) = \argmax_{i\in \{\kdp,\kddp\}, 
j\in\{\kdm,\kddm\}} \{\operatorname{sim}(a_i^+ - a_j^-, n_k)\},
\end{equation}
where $\operatorname{sim}$ is the 
cosine similarity.
The $k^+$th positive and $k^-$th negative cell form the matching pair.

\textbf{Assignment Stabilization.}
Empirically, we observe that two linear functions can overlap during the fitting (see nearly coincident sites in Fig.~\ref{fig:illustrative}), so that a small gradient oscillation can alter the matching assignment, causing instability. Therefore, we set $k^\pm:=\kdpm$ when the angle difference between $a_{k^+}^+-a_{k^-}^-$ and $a_{\kdp}^+-a_{\kdm}^-$ is less than a stabilization threshold $\theta$, which we default to $5^\circ$.

\textbf{Cell Fitting}. We introduce the loss functions and update parameters $a_{k^\pm}^\pm$ and $c_{k^\pm}^\pm$ as if the maxima in~(\ref{eq:patchwork-defintion}) were attained at the indices $k^\pm$:
\vspace{-5pt}
\begin{equation}
    \mathcal{L}_{\text{sur}} 
    = \frac{1}{K} \sum_{k=1}^K |l_{k^+}^+(x_k)-l_{k^-}^-(x_k)|, \qquad
    \mathcal{L}_{\text{normal}} = \frac{1}{K} \sum_{k=1}^K (1 - \operatorname{sim}(a_{k^+}^+ - a_{k^-}^-, n_k)).
\end{equation}
\vspace{-10pt}

We then encourage the matched pair of cells to cover $x_k$ and become the true maxima using the loss
\begin{equation}
    \mathcal{L}_{\text{cover}} = \frac{1}{K} \sum_{k=1}^K \left(\text{ReLU}\left(l_{\kdp}^+(x_k)-l_{k^+}^+(x_k)\right) 
    +\text{ReLU}\left(l_{\kdm}^-(x_k)-l_{k^-}^-(x_k)\right)
     \right).
\end{equation}

\begin{figure}
    \centering
    \includegraphics[width=1\linewidth]{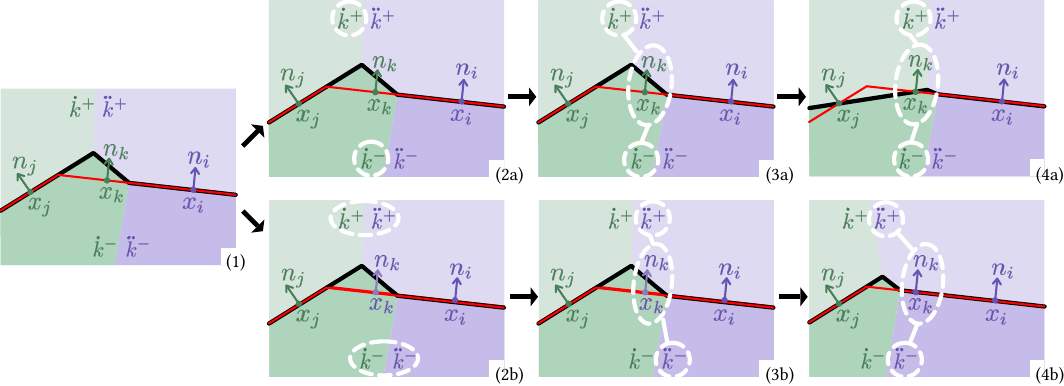}
    \vspace{-10pt}
    \caption{Discrete matching illustrated in 4 substeps, where the red line indicates the target, the black line indicates the current zero level set, and the light/dark green and light/dark purple regions refer to parts of the $\kdpm$th and $\kddpm$th positive/negative cells, respectively. (1) For the constraint $(x_{k}, n_{k})$, the standard approach only (2a) identifies and (3a) matches it to the $\kdp$th positive and $\kdm$th negative cells (the ones that contain $x_k$). (4a) This rotates the zero isosurface normal to satisfy both $n_j$ and $n_k$, causing it to deviate from the target. In contrast, our method (2b) also identifies neighboring cells and (3b) matches the constraint $(x_{k}, n_{k})$ to $\kddp$th positive and $\kddm$th negative cells, (4b) reducing the approximation error while maintaining the correct zero isosurface normal.}
    \label{fig:discrete-matching}
    \vspace{-10pt}
\end{figure}

\textbf{Pruning.}
During optimization, some cells may become empty or drift away from the surface as floaters. In either case, they no longer contribute to the active surface patches and can be removed freely. We pair each linear function with an auxiliary variable $s_i^\pm \in (0, 1]$ measuring importance and remove those whose $s_i^\pm$ falls below $0.5$ during optimization. 
We initialize all importance to be 1 and update $s_i^\pm$ via the straight-through gradient estimator~\citep{bengio2013estimating} using the loss

\vspace{-20pt}
\begin{equation}\label{eq-L-prune}
    \mathcal{L}_{\text{prune}} = \frac{1}{K} \sum_{k=1}^K \left(1 - s_{k^+}^+ + S^-_k
    + s_{\widetilde{k}^+}^++ s_{\widetilde{k}^-}^- \right),
    \quad
    S^-_k=
    \begin{cases}
        1 - s_{k^-}^-, &\text{if }\operatorname{sim}(a_{k^+}^+ \!-\! a_{k^-}^-, n_k) \ge 0,\\
        s_{k^-}^-,       &\text{if }\operatorname{sim}(a_{k^+}^+ \!-\! a_{k^-}^-, n_k) < 0,
    \end{cases}
\end{equation}
\vspace{-15pt}

where $\widetilde{k}^\pm$ denotes the element of the set $\{\kdpm,\kddpm\}$ distinct from ${k}^\pm$.
After the optimization, we only keep linear functions with $s_i^\pm > 0.5$.

\textbf{Normal Inversion Handling.}
In initializing a given shape, we offset all surface samples by the same $\delta$, regardless of local feature size within the shape. For thorny features, this can cause negative cells to overshoot the surface and extend outwards unbounded. Although this can be penalized with $\mathcal{L}_{\text{normal}}$, the cosine similarity would rotate two opposing cells by 180$^\circ$, causing large structural changes and destabilizing the fitting process. This is the reason for the second case in Eq.~(\ref{eq-L-prune}).

\textbf{Implementation details}
For each shape, we normalize its bounding box to $[-1, 1]^3$ and uniformly sample $N$ point-normal pairs for initialization, with $N$ depending on the experiment setup. During fitting, we additionally sample $K$ point-normal pairs per iteration, and minimize the total loss
$$
\mathcal{L}_{\text{total}} = 10 \cdot \mathcal{L}_{\text{sur}} + \mathcal{L}_{\text{normal}} + 100 \cdot \mathcal{L}_{\text{cover}} + \mathcal{L}_{\text{prune}}
$$
using Adam~\citep{kingma2017adammethodstochasticoptimization} over $10000$ iterations. We set the learning rate to $3 \times 10^{-4} / \tau$, with step decay of factor $0.25$ for every $40\%$ of iterations. We implement the training code in PyTorch~\citep{paszke2019pytorch} and implement the global scan of maximum/second-maximum as a fused kernel~\citep{dao2022flashattention, zhang2025efunc} using Triton~\citep{triton2019}. We also implement a variant as a 4D KD-tree using cudaKDTree~\citep{wald2023gpufriendlyparallelalmostinplaceconstruction}. We run all experiments on a single RTX3090 GPU (24 GB VRAM) and report performance in Tab.~\ref{tab:performance}. We extract our meshes via explicit polygonalization (See Appendix~\ref{sec:explicit-extraction} for details). Our full implementation and model checkpoints are available at: \url{https://github.com/Ankbzpx/patchwork-experiment}.

%% file: sec/6_results.tex
\section{Experiments}
\label{sec:experiments}

\subsection{Surface fitting}
\label{sec:surface-fitting}

We assess the expressiveness of patchwork on the surface fitting task over two shape sets spanning three well-established datasets, with increasing model capacity to match the shape complexity.

\begin{table}[t]
  \caption{Quantitative results on ABC / Thingi10K. The bold text indicates the best score; the underlined text indicates the second-best. `Ours \#Params' 
  is the average parameter size after pruning. }
  \vspace{-10pt}
  \label{tab:abc_thingi10k}
  \begin{center}
  \begin{tabular}{lcccccc}
    \bf Method & \bf CH $\downarrow$ & \bf HD $\downarrow$ & \bf FS $\uparrow$ & \bf NC $\uparrow$ & \bf IoU $\uparrow$ & \bf \#Params \\ \hline \\
    MC & 4.43 / 3.55 & 5.77 / 6.72 & 47.86 / 43.82 & 92.96 / 94.99 & 89.97 / 93.88 & 132651 \\
    RFTA & 3.73 / 3.31 & 5.67 / 5.41 & 53.10 / 52.88 & 94.47 / 96.10 & 91.55 / 94.68 & 132651 \\
    SPSR & \underline{2.27} / \underline{1.88} & \textbf{3.31} / \textbf{4.06} & 60.43 / \underline{65.71} & 96.87 / 97.82 & 95.36 / 97.59 & 131070 \\
    PoNQ & 2.41 / 2.14 & 4.38 / 4.25 & 63.91 / 60.22 & 96.70 / 96.90 & 95.89 / 97.53 & 131072 \\
    IGR & 3.07 / 1.94 & 5.17 / \underline{4.08} & 64.19 / 58.60 & 97.75 / \textbf{98.95} & 95.16 / \underline{98.06} & 132094 \\
    SIREN & 2.79 / 3.09 & 21.7 / 31.4 & 57.50 / 61.81 & \textbf{98.45} / \underline{98.73} & 95.61 / 97.94 & 132865 \\
    VoroMesh & 2.94 / 2.21 & 4.92 / 5.12 & \underline{64.39} / 59.55 & 96.65 / 95.77 & \underline{96.21} / 97.37 & 128625 \\
    DCCVT & 2.84 / 2.62 & 4.45 / 6.02 & 55.86 / 59.81 & 91.56 / 91.61 & 77.76 / 84.71 & 128625 \\
    Ours init & 5.68 / 5.48 & 20.8 / 27.6  & 58.20 / 64.10 & 87.31 / 90.53 & 92.64 / 95.49 & 131072 \\
    Ours & \textbf{1.90} / \textbf{1.80} & \underline{3.81} / 6.78 & \textbf{73.22} / \textbf{77.27} & \underline{98.05} / 98.17 & \textbf{98.39} / \textbf{99.14} & 45192 \\
  \end{tabular}
  \end{center}
  \vspace{-20pt}
\end{table}

\textbf{Datasets} ABC and Thingi10k~\citep{koch2019abc, zhou2016thingi10k}; the former is a large collection of CAD models, while the latter consists of 10K more general shapes. We pick 100 models from each, following \citep{Erler2020Points2Surf, huang2023nksr}, as our first shape set. ShapeNet~\citep{chang2015shapenetinformationrich3dmodel} is a large-scale repository of CAD shapes.
Following Neural Splines \citep{williams2021neuralsplinesfitting3d}, we use the same 260-shape subset as our second shape set, consisting of 20 shapes per category across 13 categories.
We follow CLAY~\citep{zhang2024claycontrollablelargescalegenerative} and Dora~\citep{Chen_2025_Dora}
to preprocess each shape to be manifold~\citep{wei2025neumanifold}.

\textbf{Model capacity.}
We initialize patchwork with $N = 16384$ and $65536$ surface samples across two shape sets, set $K=16384$ for all and $\tau$ to $1.0$ and $2.0$, respectively.

\begin{figure}
    \centering
    \includegraphics[width=1\linewidth]{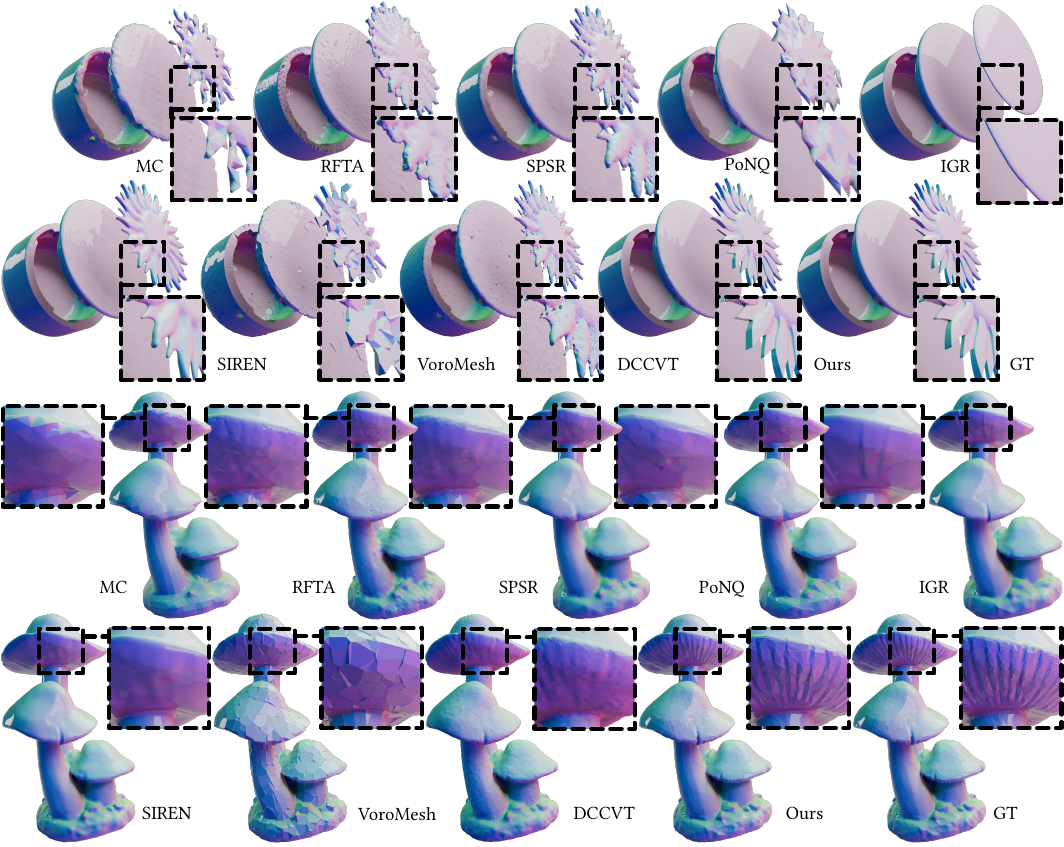}
    \vspace{-15pt}
    \caption{Visual results on ABC~\citep{koch2019abc} and Thingi10k~\citep{zhou2016thingi10k}.}
    \label{fig:abc-thingi10k}
    \vspace{-10pt}
\end{figure}

\textbf{Baselines} We compare against two foam-based representations, VoroMesh~\citep{maruani2023voromesh} and DCCVT~\citep{charawi20263dv}, and three common topology-agnostic families: regular grids, evaluated via MC~\citep{lorensen1987mc} and RFTA~\citep{Sellan2024RFTA}; point sets, evaluated via SPSR~\citep{kazhdan2013spsr} and PoNQ~\citep{maruani2024ponq}; and neural implicit, via IGR~\citep{gropp2020implicitgeometricregularizationlearning} and SIREN~\citep{sitzmann2020implicit}. Since our final parameter count varies across shapes due to 
pruning, we match all baselines to the parameter count at our initialization (i.e., $8N$, excluding auxiliary $s^{\pm}$), the \emph{upper bound} of our final raw floating-point storage size on disk. For fitting-based methods (PoNQ, IGR, SIREN, VoroMesh,  DCCVT), we resample $K=16384$ surface samples each iteration to match our protocol; exact configurations and fitting details are given in Appendix~\ref{sec:baseline-details}. We also 
report our final parameter counts after pruning, averaged within each shape set.

\textbf{Metrics} For quantitative evaluation, we normalize all results in range $[-1, 1]$, sample 1M points on both ground truth and raw output surfaces and report Chamfer Distance (CH, $\times 10^{-3}$)~\citep{fan2016pointsetgenerationnetwork}, Hausdorff Distance (HD, $\times 10^{-2}$), F-score (FS)~\citep{Knapitsch2017} with a cutoff $0.1\%$ largest dimension of the bounding box, Normal Consistency (NC, $\times 10^{-2}$) and Intersection-over-Union (IoU, $\times 10^{-2}$)~\citep{peng2020convolutional}. For IGR and SIREN, we use MC of resolution $512^3$ to extract the underlying surfaces. For DCCVT, we correct the output surface normal based on its auxiliary HotSpot SDF~\citep{wang2024hotspotscreenedpoissonequation} gradient.

\begin{table}[t]
  \caption{Quantitative results on the ShapeNet dataset. The bold text indicates the best score; the underlined text indicates the second-best. `Ours \#Params' 
  is the average parameter size after pruning. }
  \label{tab:shapenet}
  \begin{center}
  \vspace{-10pt}
  \begin{tabular}{lclccccc}
    \bf Method & \bf CH $\downarrow$ & \bf HD $\downarrow$ & \bf FS $\uparrow$ & \bf NC $\uparrow$ & \bf IoU $\uparrow$ & \bf \#Params \\ \hline \\
    MC & 29.2 & 26.7 & 17.63 & 87.36 & 15.39 & 531441 \\
    RFTA & 16.7 & 17.3 & 26.45 & 88.67 & 27.24 & 531441 \\
    SPSR & 3.09 & \textbf{3.10} & 33.75 & 86.89 & 54.64 & 524286 \\
    PoNQ & 3.48 & 9.13 & 42.43 & 85.08 & 58.10 & 524288 \\
    IGR & 4.70 &	11.1 &	\underline{54.85} &	94.69 &	\underline{70.99} & 526334 \\
    SIREN & \underline{2.66} & 8.37 & 49.82 & \textbf{97.89} & 66.08 & 527873 \\
    VoroMesh & 63.3 & 45.6 & 8.03 & 62.43 & 10.41 & 526848 \\
    DCCVT & 3.52 & 4.73 & 24.93 & 68.04 & 24.17 & 526848 \\
    Ours init & 6.56 & 24.4 & 30.44 & 68.00 & 47.80 & 524288 \\
    Ours & \textbf{1.65} & \underline{3.28} & \textbf{64.26} & \underline{97.77} & \textbf{96.69} & 90808 \\
  \end{tabular}
  \end{center}
  \vspace{-10pt}
\end{table}

\textbf{Quantitative Results} Patchwork achieves the best Chamfer Distance, F-Score, and IoU across all datasets. Our F-score consistently exceeds the second-best method by roughly 10 percentage points, demonstrating both higher accuracy and completeness~\citep{Knapitsch2017}. Moreover, it does so with only $ 34\%$ and $ 17\%$ of the matched parameter budget after pruning.

\textbf{Qualitative Results} Compared to VoroMesh, patchwork substantially outperforms it on smoother, more complex shapes, with the gap widening as shape complexity increases. Compared to DCCVT, patchwork outperforms it at representing sharp creases for CAD models, and scales better for more complex shapes under the same parameter budget. Compared to IGR, patchwork represents intricate geometric detail more faithfully; compared to SIREN, our hard maximum preserves sharper creases. More visual results are available in Appendix~\ref{sec:extended-visual-results}.

\begin{figure}[b]
    \vspace{-20pt}
    \centering
    \includegraphics[width=1\linewidth]{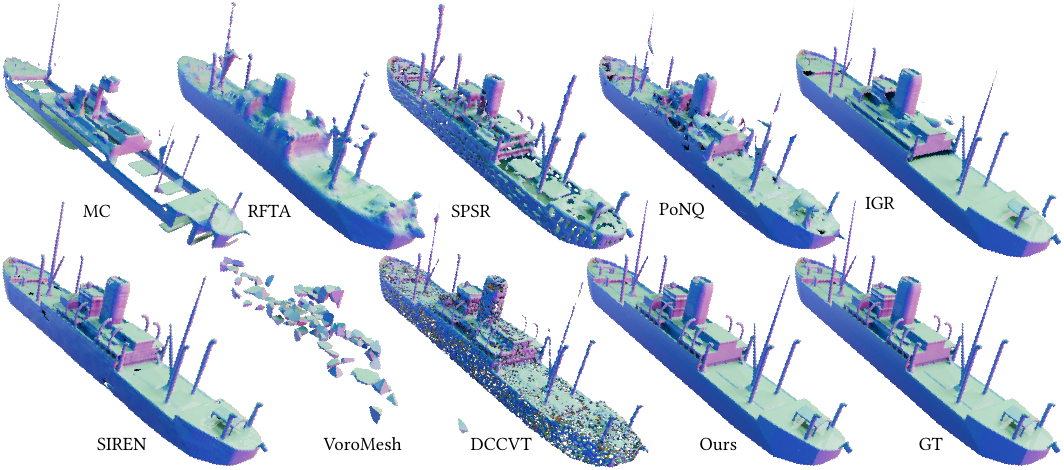}
    \vspace{-15pt}
    \caption{Qualitative results on ShapeNet~\citep{chang2015shapenetinformationrich3dmodel}.}
    \label{fig:shapenet}
\end{figure}

\subsection{Scalability}
To examine scalability, we compare against GeomDist~\citep{Zhang_2025_ICCV}, which encodes shapes as distributions via a diffusion model and achieves SOTA geometric fidelity, surpassing even hash-grid-based methods \citep{muller2022instant}. Our method faithfully captures outer surfaces for all cases and offers explicit polygonalization, achieving higher accuracy on 7 of 10 shapes (mean Chamfer Distance of 2.380 $\times10^{-3}$ vs. 3.702 $\times10^{-3}$), but falls short on the remaining 3 in handling self-intersection and non-manifold interior parts, resulting in spurious geometry enclosed within the outer surface. Our results are visualized in Fig.~\ref{fig:teaser}; details are supplied in Appendix~\ref{sec:baseline-details}.

\subsection{Ablation studies}
We ablate our major design choices to isolate their contributions to our final performance (Tab.~\ref{tab:ablation}). Specifically, we study discrete matching, assignment stabilization, geometric initialization, progressive pruning, normal inversion handling, learning rate scheduler, and the hard maximum on ABC/thingi10k datasets. 

\textbf{Discrete Matching (DM).} We ablate the importance of discrete matching against picking the maximum pair. Without discrete matching, the optimization leads to low-polygonal approximations, with greatly reduced accuracy as indicated by F-score, and diminished visual details (Fig. \ref{fig:ablation}).

\textbf{Assignment Stabilization (AS).} To assess the impact of assignment stabilization, we keep the matched maximum assignment regardless of how small their angle differences are. This configuration is prone to gradient noise, such that the pruning set is hard to stabilize, resulting in significantly more floators, as indicated by a substantially higher Hausdorff Distance.

\textbf{Geometric Initialization (GI).} We replace the geometric initialization with random uniform initialization. Although it can fit to some degree, the foam tessellation is significantly sparse. Moreover, it is incompatible with our pruning loss, such that the resulting scene is populated with floators.

\textbf{Pruning (P).} We disable pruning, so that polyhedra that fall out of the surface remain present, yielding a higher Hausdorff Distance.

\textbf{Normal Inversion Handling (NIH).} We remove the Normal Inversion Handling and observe that some outwards-extended inner polyhedra remain unpruned after the optimization.

\textbf{Learning Rate Scheduler (LRS).} In our formulation, the difference of power distance is subject to a constant shift, so the learning rate scheduler anneals the oscillation in that null space. Moreover, it also stabilizes the maximum assignment, the pruning set, towards the end of optimization. Removing it introduces gaps between patches and floators off the surface, degrading the visual quality.

\textbf{Hard Maximum (HM).} Finally, we ablate hard maximum by replacing it with LogSumExp, its smooth approximation. Such substitution requires different techniques; we report the best-performing one, with details in Appendix~\ref{sec:implicit-fitting}. Our hard maximum with combined techniques improves the Chamfer Distance and F-score, and significantly boosts the visual fidelity.

\begin{table}[t]
  \caption{Ablation studies of Discrete Matching (DM), Assignment Stabilization (AS), Geometric Initialization (GI), Pruning (P), Normal Inversion Handling (NIH), Learning Rate Scheduler (LRS), and Hard Maximum (HM) on the ABC / Thingi10k dataset. The bold text indicates the best score.}
  \label{tab:ablation}
  \begin{center}
  \vspace{-10pt}
  \begin{tabular}{lclccccc}
      & \bf CH $\downarrow$ & \bf HD $\downarrow$ & \bf FS $\uparrow$ & \bf NC $\uparrow$ & \bf IoU $\uparrow$ \\ \hline \\
    w.o. DM & 3.74 / 4.17 & 5.77 / 9.32 & 37.78 / 31.85 & 94.90 / 94.49 & 92.24 / 94.34 \\
    w.o. AS & 6.15 / 3.15 & 33.9 / 13.3 & 66.98 / 74.74 & 93.64 / 97.44 & 95.46 / 97.84 \\
    w.o. GI & 141 / 106 & 102 / 78.1 & 9.99 / 10.1 & 56.48 / 57.83 & 21.32 / 33.71 \\
    w.o. P & 4.92 / 3.84 & 24.5 / 22.7 & 71.27 / 76.27 & 96.23 / 97.39 & 95.19 / 97.28 \\
    w.o. NIH & 2.08 / 2.51 & 4.55 / 8.13 & 72.64 / 76.74 & 97.89 / \textbf{98.17} & 97.96 / 98.05 \\
    w.o. LRS & 2.17 / 2.96 & 7.53 / 9.35 & 72.05 / 74.80 & 97.94 / 97.89 & 97.22 / 97.16 \\
    w.o. HM & 2.48 / 2.39 & 4.93 / \textbf{5.88} & 60.83 / 61.59 & 97.55 / 97.26 & 96.05 / 97.34 \\
    Full & \textbf{1.90} / 
    \textbf{1.80} & \textbf{3.81} / 6.78 & \textbf{73.22} / \textbf{77.27} & \textbf{98.05} / \textbf{98.17} & \textbf{98.39} / \textbf{99.14} \\
  \end{tabular}
  \end{center}
  \vspace{-10pt}
\end{table}

%% file: sec/5_conclusions.tex
\begin{figure}
    \vspace{-20pt}
    \centering
    \includegraphics[width=1\linewidth]{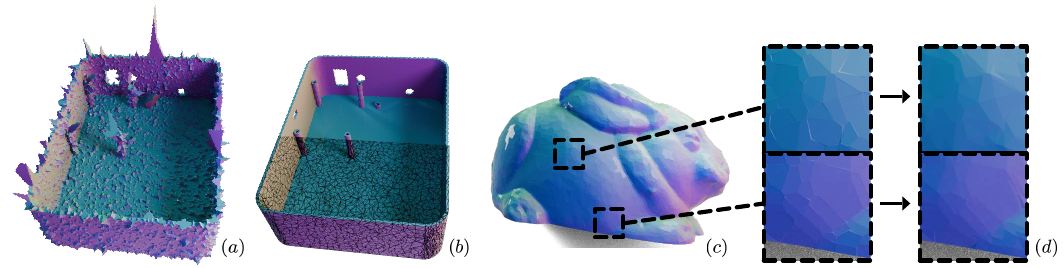}
    \vspace{-15pt}
    \caption{(a) Our geometric initialization, (b) optimization result (top) and its tessellation (bottom). (c) The raw output of our explicit polygonalization can have minor seams; (d) these can mostly be dissolved with our edge collapsing routine.}
    \label{fig:geo-init}
    \vspace{-15pt}
\end{figure}

\begin{figure}[h]
    \centering
    \includegraphics[width=1\linewidth]{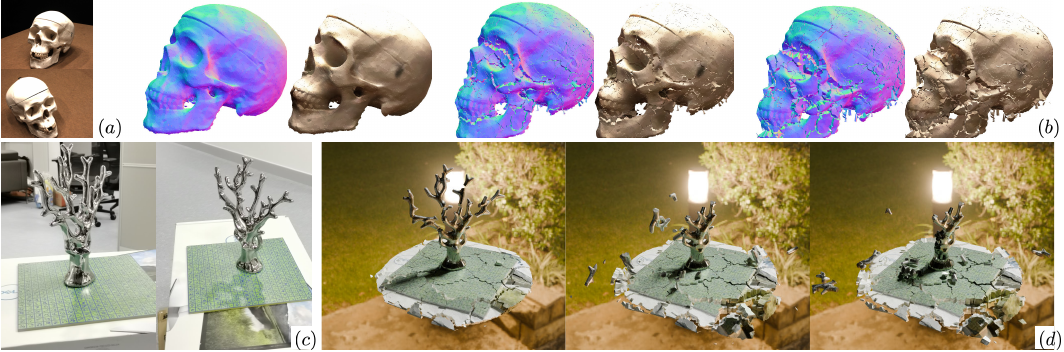}
    \vspace{-15pt}
    \caption{Fracture simulation. Our representation is fitted from pretrained neural rendering models on multi-view images (a, c); fracture sequences are rendered via (b) Neural Deferred Shading~\citep{worchel:2022:nds} and (d) relit with inverse-rendered BRDF materials~\citep{liu2023nero}.}
    \label{fig:application}
\end{figure}

\section{Applications}
Beyond typical applications of implicit volumetric representations, such as compression, ray tracing, and inside/outside classification (Fig.~\ref{fig:teaser}), our method additionally supports explicit extraction of conforming volumetric meshes of convex volume elements, enabling direct downstream tasks like fracture simulation (Fig.~\ref{fig:application}).

\section{Limitations and Future work}
\label{sec:limitation}

\textbf{Limitations.} As in related lines of work (Tab.~\ref{tab:feature-table}), our optimization landscape is sensitive to local minima and requires careful initialization. Our initialization may produce incorrectly extended polyhedra that are not pruned due to insufficient sampling or erroneous normal constraints (e.g., self-intersections in the GT mesh). We observe both cases in the Thingi10K dataset, as indicated by our higher Hausdorff Distance. Our method requires clean constraints and is sensitive to noise (Fig.~\ref{fig:noise} and Appendix~\ref{sec:noise-tolerance}). On curved surfaces, seams can appear when our tessellation does not fully align with the ground truth mesh. This is a historical challenge dating back to Variational Surface Approximation~\citep{cohen2004variational}. For visual results, we mitigate this with a custom edge-collapse routine (Fig.~\ref{fig:geo-init}~(d)). We expect generalizing the cocircularity constraint of concurrent work VoroLight~\citep{lu2026vorolightlearningvoronoisurface} to our setting would further improve visual quality.

\textbf{Future work.}
Our current optimization does not guarantee a theoretically minimal configuration and is limited to closed manifold surfaces. In principle, however, inter-polyhedral facets of foams can exhibit more diverse topologies, including open and non-manifold surfaces. We leave them, together with adaptive polyhedra placement and densification extensions, for future work (see also Appendix~\ref{sec:theoretical-future-work}). Finally, given the increasing interest in part-level shape understanding via volumetric feature fields~\citep{partfield2025, zhu2026partsam}, our representation may serve as an effective medium for applications such as semantics-aware shape decomposition and fracturing.

\begin{figure}[h]
  \begin{minipage}[c]{0.48\linewidth}
    \centering
    \vspace{-10pt}
    \captionof{table}{Per-iteration performance of our Triton / KD-tree implementations with batch size $K=16384$ over varying $N$, measured on a single NVIDIA RTX3090. The parenthesized values show KD query time only.}
    \label{tab:performance}
    \vspace{-5pt}
    \resizebox{\linewidth}{!}{
    \begin{tabular}{lcc}
      $N$ & Time (s) & VRAM (MB) \\ \hline \\
    16384 & 0.0335 / 0.0341 (0.0035) & 352 / 384 \\
    65536 & 0.0375 / 0.0614 (0.0101) & 364 / 394 \\
    655360 & 0.1927 / 0.1060 (0.0392) & 520 / 620 \\
    1048576 & 0.2992 / 0.1510 (0.0395) & 616 / 794 \\
    16777216 & 4.6207 / 2.1561 (1.3151) & 4458 / 6796 \\
    \end{tabular}}
  \end{minipage}%
  \hfill
  \raisebox{10pt}{
  \begin{minipage}[t]{0.48\linewidth}
    \centering
    \includegraphics[width=1\linewidth]{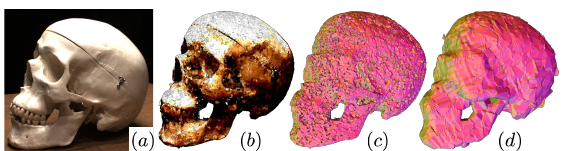}
    \vspace{-10pt}
    \captionof{figure}{Noisy constraints degrade our output to low-poly-like results. (a) Reference view. (b) Textured and normal rendering (c) of the pretrained SDFoam raw output. (d) Our fit using constraints estimated from volume density.}
    \label{fig:noise}
  \end{minipage}}
  \vspace{-15pt}
\end{figure}

\section{Conclusions}
\label{sec:conclusions}

We present patchwork, a foam-based 3D representation that explicitly encodes shapes as piecewise polyhedra while exhibiting a simple implicit form. We formally prove its representation bounds and approximation capability, and introduce a combinatorial matching scheme that enables fitting patchwork to its theoretical potential. Extensive comparisons against existing baselines under carefully matched parameter counts show that patchwork surpasses them while requiring 3--6$\times$ fewer parameters on average. We further demonstrate its unique ability to extract conforming volume meshes, enabling direct downstream fracture simulation. We envision patchwork becoming a valuable foundation for geometric learning.

%% file: sec/6_appendix.tex
\appendix


\section{Origins of the method}
\label{sec:origin}

Originally, patchwork representation for curves appeared in pure mathematics as a constructive approach to famous \textit{Hilbert's sixteenth problem}. This problem was stated in the milestone talk ``Mathematical Problems'' by David Hilbert at the Second International Congress of Mathematicians in 1900. The first part of this problem asked for a possible number and mutual position of components of a curve represented by an algebraic equation of a given degree. This question remains widely open, starting from degree $8$. It further decomposes into two sub-tasks: to \textit{construct} an algebraic equation that represents a curve of given topological type and to \textit{bound} the complexity of algebraic curves of given degree. 

\textit{Patchworking} is an approach to this construction, which traces back to the works by Gudkov~\citep{gudkov1974topology} from the 1970s and was later formalized and generalized by Viro~\citep{viro1986progress}. See an elementary introduction in~\citep{IV-96}. Using this approach, a counterexample to the original Hilbert's result was constructed (seventy years later), as well as to other famous conjectures in algebraic geometry, and Hilbert's problem was solved for equations of degree up to $7$. Furthermore, it served as a motivation and foundation to \textit{tropical geometry}, which studies piecewise-linear analogs of algebraic curves and surfaces~\citep{maclagan2015introduction}. The underlying
formalism has much similarity with neural networks,
so it is natural to build a bridge from patchworks to neural shape representations; cf.~\citep{brandenburg2024real}.

The \textit{complexity bounds} for algebraic curves are obtained in \textit{real algebraic geometry} by very different methods~\citep{Yuan-14}. An example is Harnack's curve theorem, which bounds the number of components of an algebraic curve of given degree. Such results also bound the representation capacity of patchworks. 
For algebraic surfaces in space, a tight bound is unknown starting from degree $5$; see \citep{itenberg2006asymptotically}.

Our approach is directly inspired by the mathematical concept of patchworking \citep{viro1986progress,IV-96,itenberg2006asymptotically} and tropical geometry~\citep{maclagan2015introduction}.
For applications of tropical geometry for deep learning, see~\citep{alfarra2022decision,brandenburg2024real,maragos2021tropical,zhang2018tropical}.

\section{Introduction to the method}
\label{sec:method}

In this supplementary material, we introduce our main concept step by step. We start with a concise 
definition in 2D, with a discussion of its visualization and expressive power. Then we generalize it to~3D, rephrasing Definition~\ref{def-patchwork-general} in an equivalent form, more suitable for theoretical analysis.

\subsection{2D}
We represent a curve in the plane as the zero level set of a field $f\colon\mathbb{R}^2 \to \mathbb{R}$ of a certain special form.

\begin{definition}[See Figs.~\ref{fig-patchwork}--\ref{fig-patchwork-examples}]
A \emph{2D patchwork field} is a function $\mathbb{R}^2\to\mathbb{R}$ of the form
\begin{equation}
\label{eq-def-patchwork-limiting}
f(x,y)=\max_{\substack{1\le i\le n:\\s_i=+1}}\left\{a_{i}x+b_{i}y+c_{i}\right\}-\max_{\substack{1\le i\le n:\\s_i=-1}}\left\{a_{i}x+b_{i}y+c_{i}\right\},
\end{equation}
depending on the parameters $n\in\mathbb{Z}_{>0}$ (called \emph{width}), 
$s_i\in \{+1,-1\}$ (
\emph{signs}), and $a_i,b_i,c_i\in \mathbb{R}$, where the first maximum is taken over indices $i\in\{1,\dots,n\}$ such that $s_i=+1$ and the second maximum is over $i$ such that $s_i=-1$. 
A \emph{patchwork} is the curve in the plane given by $f(x,y)=0$.
%
%
%
\end{definition}

\begin{figure}[htb]
    \centering
    \includegraphics[width=0.15\textwidth]{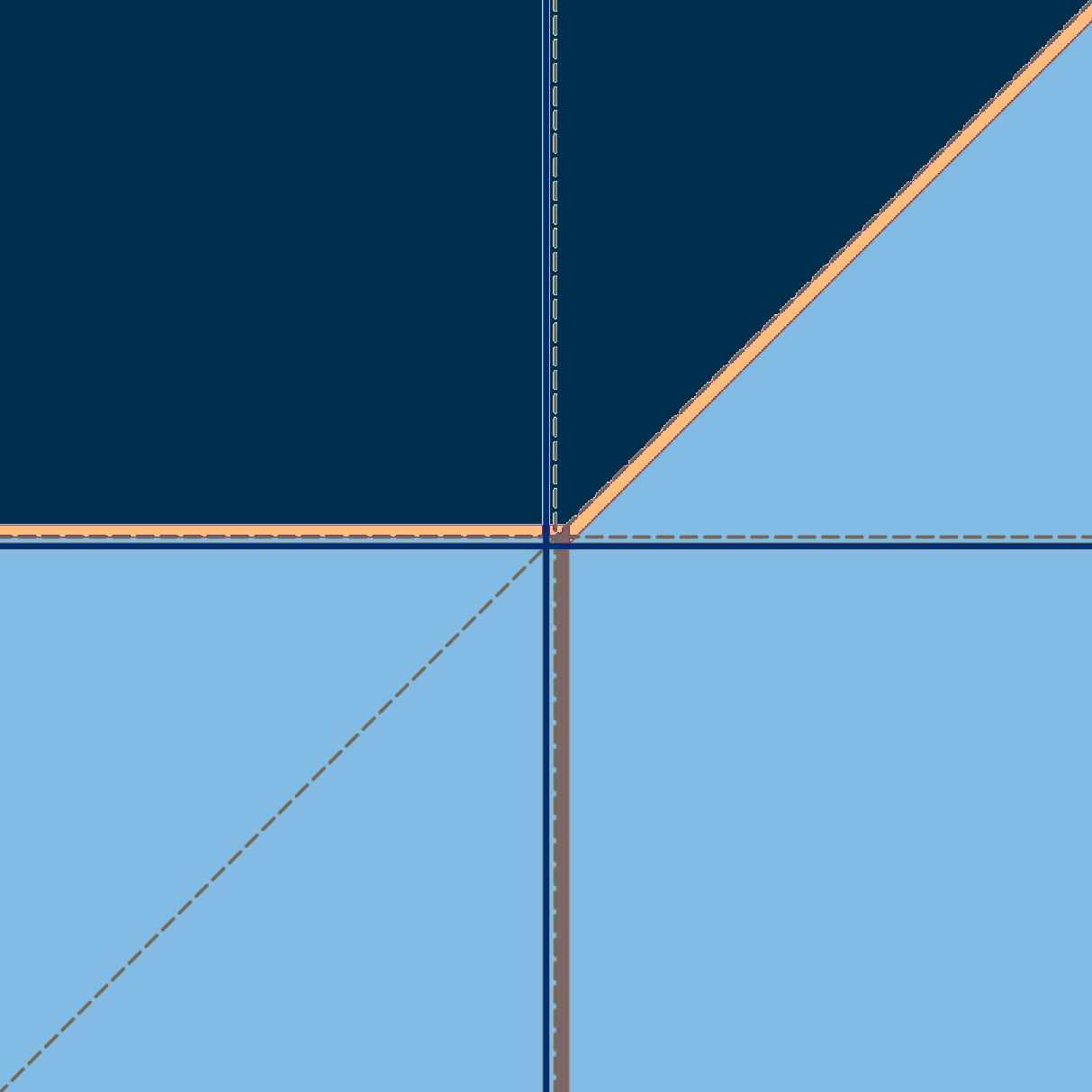}\,
    \includegraphics[width=0.15\textwidth]{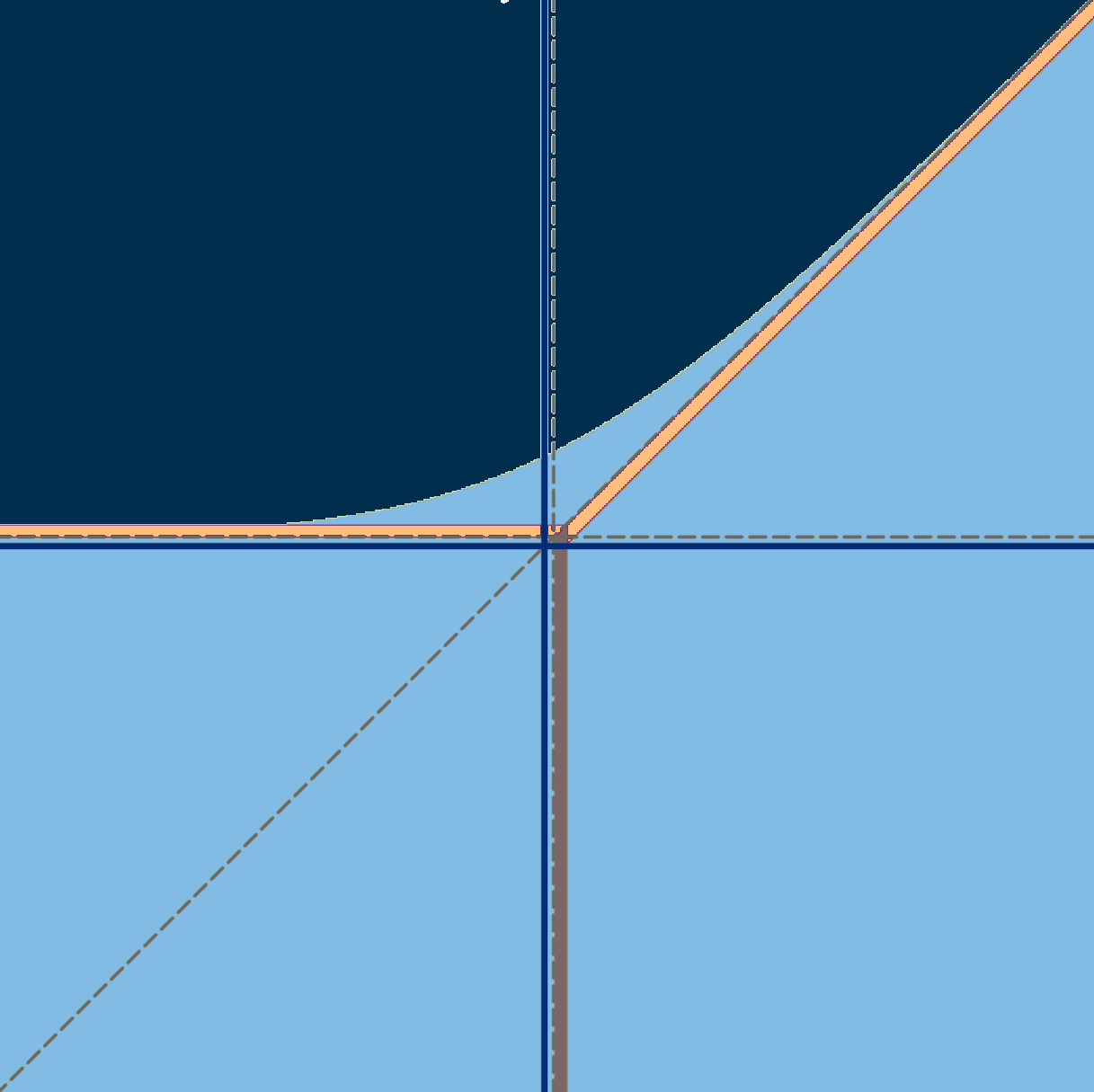}\,
    \includegraphics[width=0.15\textwidth]{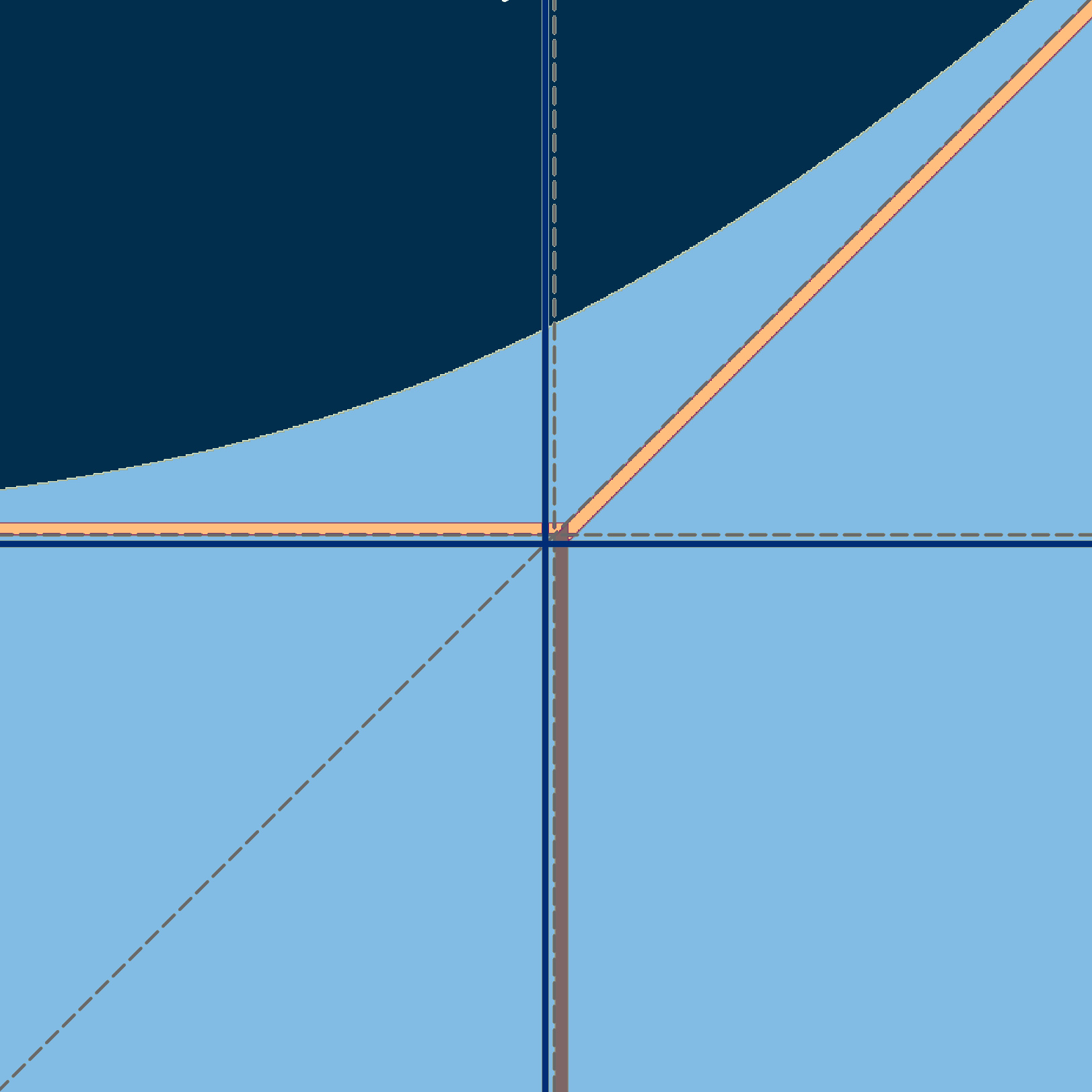}
    \caption{Left: A patchwork (the boundary between the dark and light areas).
    The design lines (blue), the equality lines (dashed), the candidate intervals (brown/orange), and the active intervals (orange). 
    Just two design lines represent a concave angle. Middle and right: smoothed patchworks of increasing smoothness $T$ (from left to right). The patchwork (left) is the limiting case $T\searrow 0$.} 
    \label{fig-patchwork}
\end{figure}

\begin{figure}[htb]
    \centering
    \includegraphics[width=0.15\textwidth]{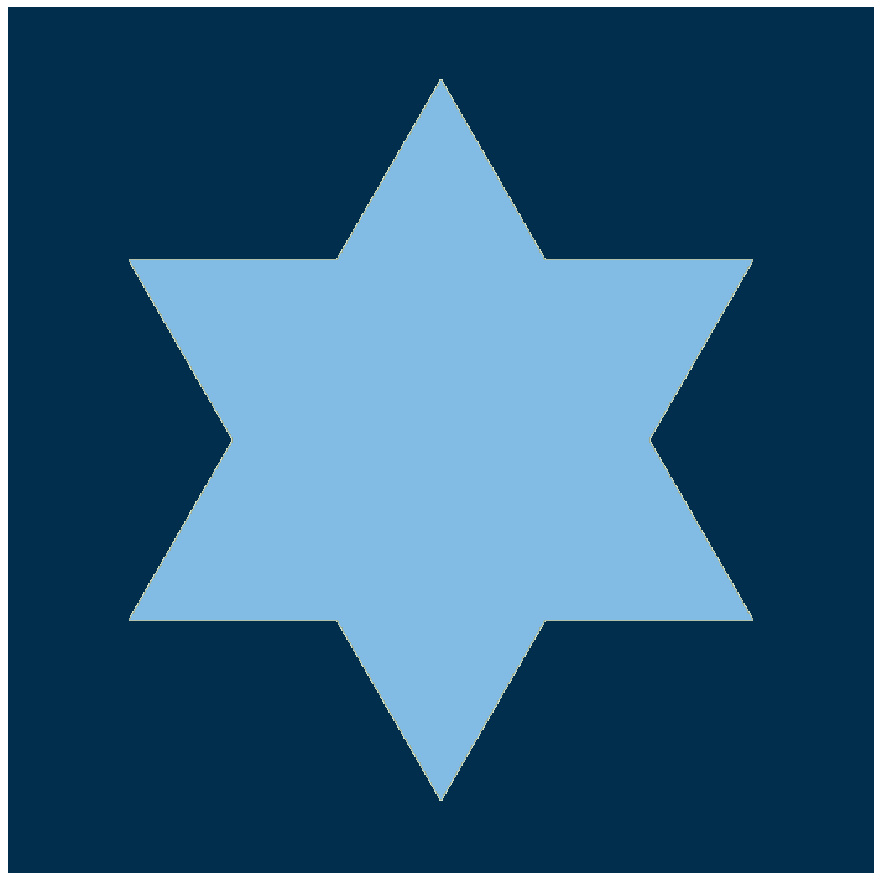}
    \includegraphics[width=0.15\textwidth]{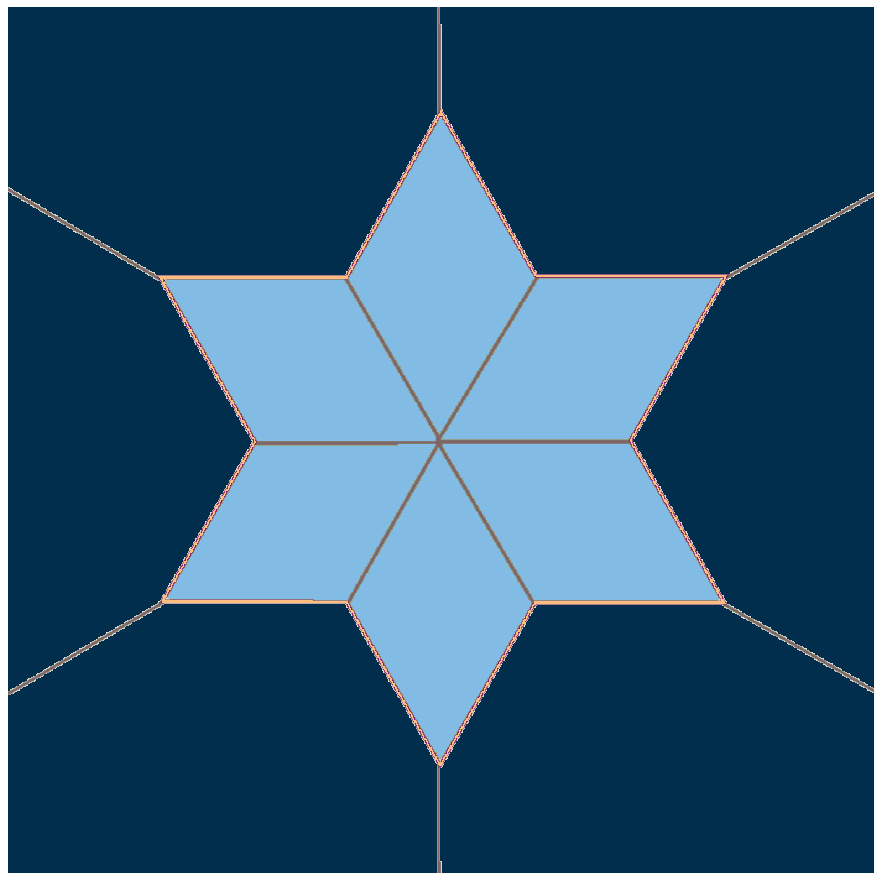}\hfill
    \includegraphics[width=0.15\textwidth]{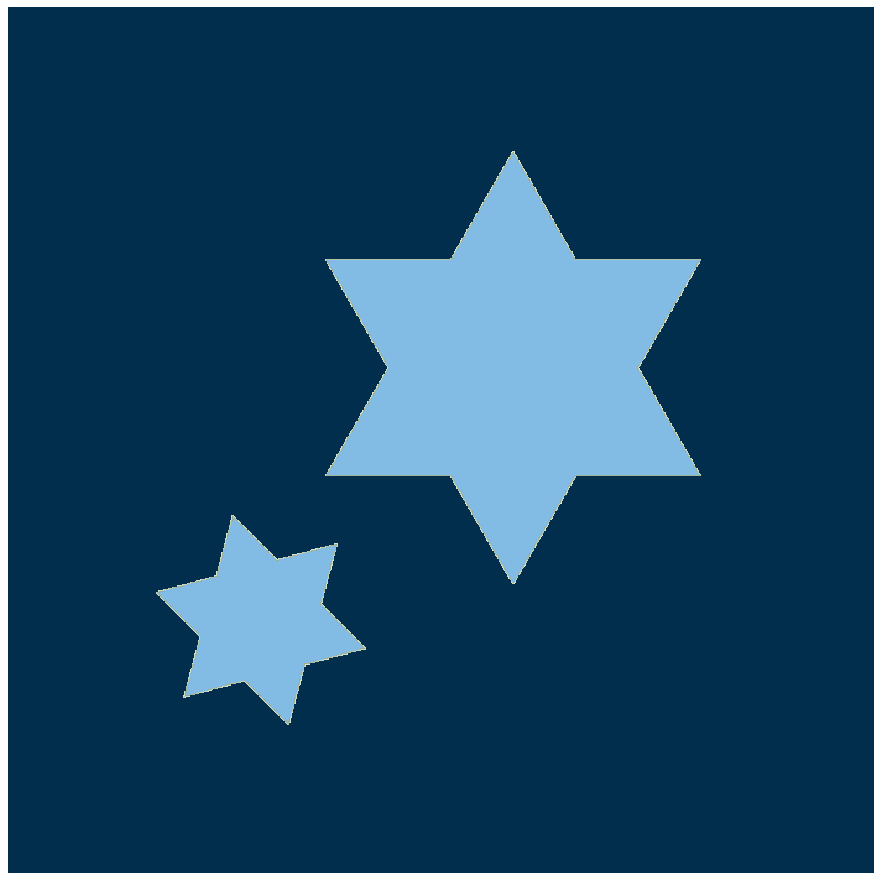}
    \includegraphics[width=0.15\textwidth]{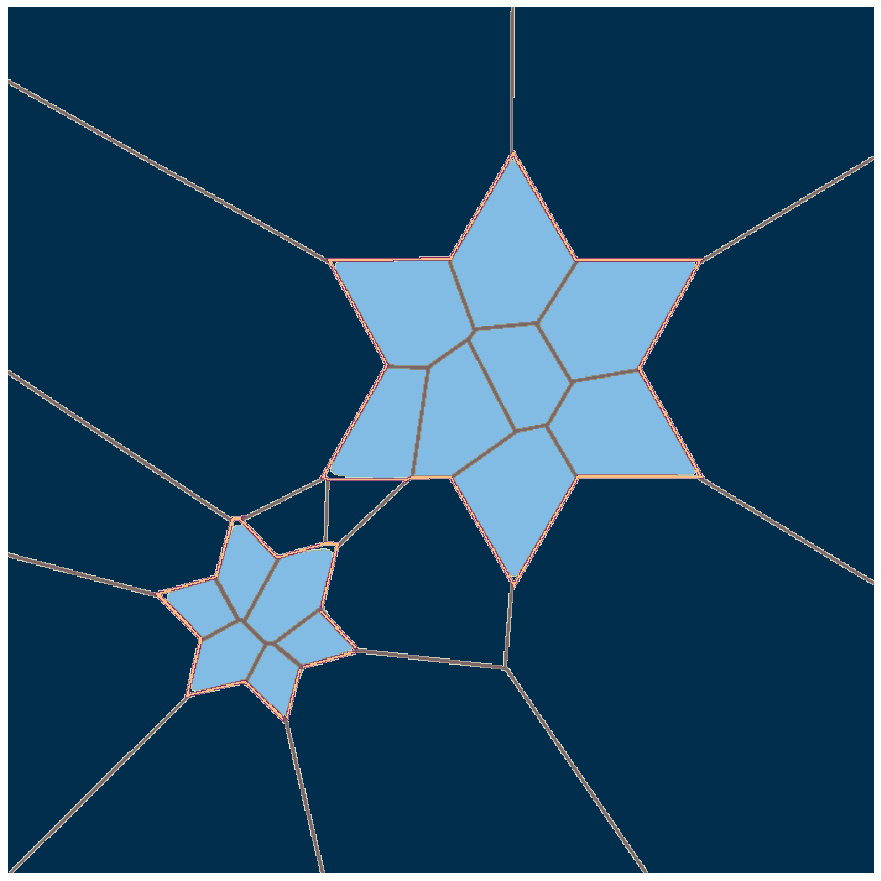}\hfill
    \includegraphics[width=0.15\textwidth]{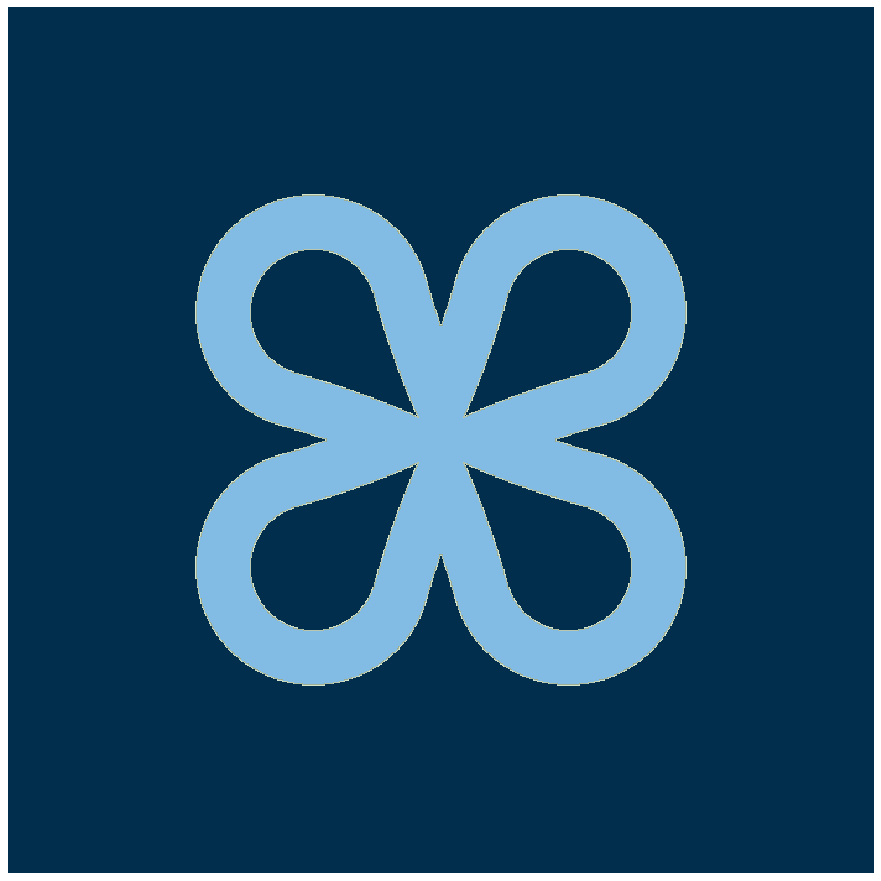}
    \includegraphics[width=0.15\textwidth]{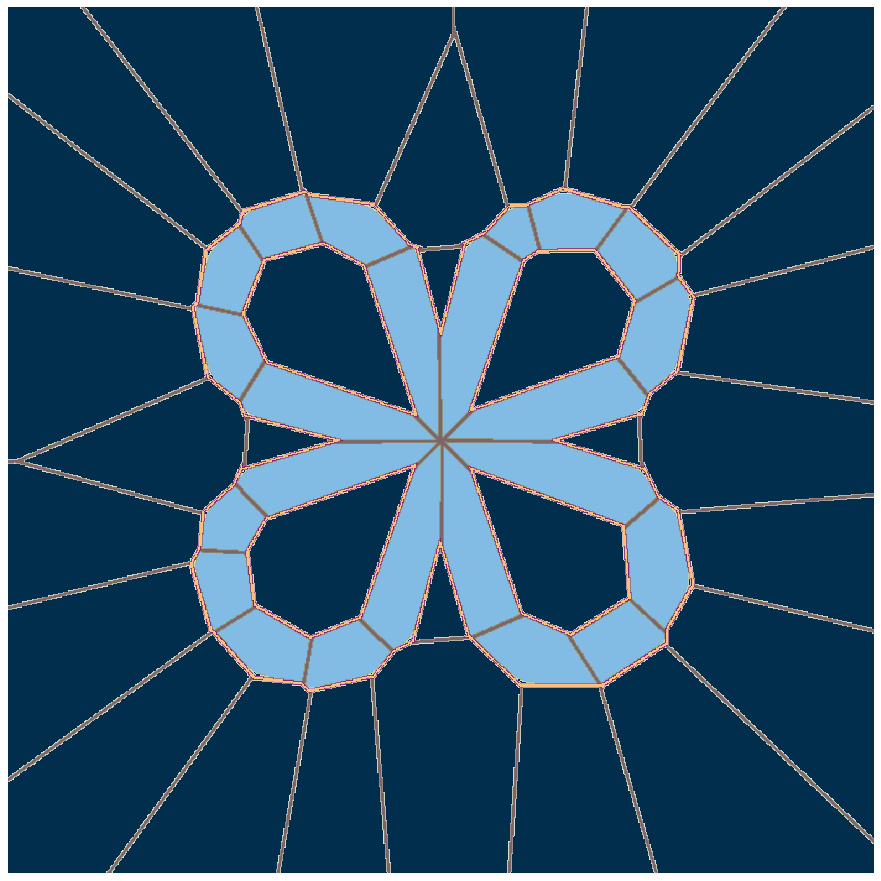}
    \caption{
    Nonconvex (left), disconnected (middle), and 
    non-simply-connected (right) shapes and their approximations by patchworks computed by optimization (right per pair). Each patchwork is the boundary between a dark and a light region. These compact representations have just $n=12, 37, 55$ design lines (left to right). 
    The candidate intervals (brown/orange) decompose the plane into polygons. The ``interior'' of the patchwork is 
    the union of the polygons corresponding to the indices $i$ with $s_i=-1$, and the patchwork is 
    the union of active intervals (orange).  
    }
\label{fig-patchwork-examples}
\end{figure}
\textbf{Visualization.}
This construction is visualized
as follows; see Fig.~\ref{fig-patchwork}. 
Consider the patchwork $f(x,y)=0$ and let us determine its shape.
Introduce the linear function
\begin{equation}\label{eq-design-lines}
l_i(x,y):=a_{i}x+b_{i}y+c_{i} \qquad\text{for each }i=1,\dots,n.
\end{equation}
Assume that the functions $l_1(x,y),\dots,l_n(x,y)$ are distinct and one of them is a positive constant.
If $l_i(x,y)$ is not a constant, then $l_i(x,y)=0$ is an equation of a line in the plane, referred to as a \emph{design line}.  For instance, in Fig.~\ref{fig-patchwork},
$l_1(x,y)=x$, $l_2(x,y)=y$, $l_3(x,y)=0.01$ so that the design lines are the coordinate axes. 

Observe that function~(\ref{eq-def-patchwork-limiting}) vanishes only if both maxima are equal, in particular, $l_i(x,y)=l_j(x,y)$ for some $i\ne j$. The lines given by the equations $l_i(x,y)=l_j(x,y)$ for all pairs of distinct indices $i,j\in\{1,\dots,n\}$ are called \emph{equality lines}. In what follows, fix such a pair $i\ne j$. For instance, in Fig.~\ref{fig-patchwork}, the equality lines are $x=y$, $x=0.01$, and $y=0.01$.

Next, a point $(x,y)$ on the equality line satisfies~$f(x,y)=0$,
only if the value $l_i(x,y)=l_j(x,y)$ is maximal among
$l_1(x,y),\dots,l_n(x,y)$. The intervals where
\begin{equation}
l_i(x,y)=l_j(x,y)\ge \max_k l_k(x,y)
\end{equation}
are called \emph{candidate intervals}. They can be straight line segments, rays, or lines. For instance, in Fig.~\ref{fig-patchwork}, these are three rays emanating in the north-east, south, and west directions. Their union is known as the \emph{tropical line}.

The candidate intervals split the plane into convex polygons, possibly empty or unbounded, given by the condition $l_i(x,y)\ge \max_{k}l_k(x,y)$ for some~$i$.
Each one is referred to as the \emph{candidate polygon} (or the \emph{patchwork cell}) \emph{corresponding to the index $i$}. 
In Fig.~\ref{fig-patchwork} on the left,
ones with $s_i=-1$ are light, and the one with $s_i=+1$ is dark.

\input{fig/hex/hex}

Finally, to satisfy~$f(x,y)=0$, the signs $s_i$ and $s_j$ must be opposite.
Candidate intervals with this property are called \emph{active intervals}.  In Fig.~\ref{fig-patchwork}, they are orange.

We conclude that the patchwork is
the union of active intervals  (if 
some $l_i(x,y)$ is a positive constant, and the parameters are in general position; see \journalvsarxiv{the supplementary materials}{Appendix~\ref{sec:dependence}}). 

\textbf{Relation to power diagrams and piecewise linear functions.} Eq.~(\ref{eq-def-patchwork-limiting}) can be rewritten in the following form by negating the linear functions and completing the squares:
\begin{equation}\label{eq-power-diagram-2D}
    f(x) = \min_{\substack{1\le i\le n:\\s_i=-1}}\{|x - p_i|^2 - w_i\} 
         - \min_{\substack{1\le i\le n:\\s_i=+1}} \{|x - p_i|^2 - w_i\},
\end{equation}
where $p_i = \tfrac{1}{2} (a_i,b_i) \in \mathbb{R}^2$ and $w_i = \tfrac{1}{4} (a_i^2+b_i^2) + c_i \in \mathbb{R}$. This form shows that the candidate polygons are exactly the cells of the power 
diagram with points $p_i$ and weights $w_i$. However, Eq.~(\ref{eq:patchwork-defintion}) is simpler and more convenient for both theoretical analysis and optimization.

Next, a patchwork field is clearly linear in each candidate polygon, hence is a piecewise-linear function. Conversely, each continuous piecewise linear function can be represented in form~(\ref{eq-def-patchwork-limiting}), thus is a patchwork field; see \citep{siahkamari2020piecewise}. However, form~(\ref{eq:patchwork-defintion}) is more convenient for us. 

\textbf{Expressive power.} Our construction allows visual control
and compact representations even of complicated non-convex shapes; see Fig.~\ref{fig-patchwork-examples}. It also suggests the following design of a patchwork approximating a given curve; see Fig.~\ref{fig:hex}.

Let the curve be inside the box $|x|,|y|\le 1$.
Superimpose a grid of step $1/N$.
Paint the squares
lying inside the
curve light, and the remaining ones 
dark. The union of the sides
separating
colors is called a \emph{digital curve}.
It can be explicitly represented as a patchwork.  
For notational convenience, replace the index $i$ 
in~(\ref{eq-def-patchwork-limiting}) with a pair of indices $k,l\in [-N,N]$, so we write $a_{k,l}$ instead of $a_{i}$ etc. Set 
\begin{equation}\label{eq-hex}
\begin{aligned}
&a_{k,l}=k, \quad 
b_{k,l}=l, \quad 
c_{k,l}=(1-k^2-l^2)/(2N), \\
&s_{k,l}=\begin{cases}
    +1,
    &\text{if the candidate polygon} \\ &\text{corresponding to the pair $k,l$ is dark,}\\
    -1, &\text{otherwise.}
\end{cases}
\end{aligned}
\end{equation}
In particular, $l_{0,0}(x,y)=1/(2N)$.
The resulting candidate polygons form a square grid inside the box (this can be shown using power diagrams or \citep[Lemma~1]{pirahmad2025}), and 
the patchwork coincides with the digital curve.
Thus, any plane curve can be approximated by patchworks.
\journalvsarxiv{
See Theorem B.1 
in the supplementary material.}{See Theorem~\ref{th-representation-property-2D} in Appendix~\ref{sec:proofs}.}
There, we also discuss hexagonal lattices (see Fig.~\ref{fig:hex} on the right; cf.~\citep[Lemma~3.8]{KSS-25}). 
We can further optimize
the representation; see Fig.~\ref{fig-patchwork-examples}.

We also introduce smoothed patchworks by replacing the maximum in~(\ref{eq-def-patchwork-limiting}) with a suitable soft-max.

\begin{definition}[See Fig.~\ref{fig-patchwork}]
\label{def-smoothed-patchwork-2d}
    A \emph{smoothed 2D patchwork field} is
\begin{equation}
\label{eq-def-patchwork}
F(x,y)=T
\log\sum_{i=1}^n s_i\exp \frac{a_{i}x+b_{i}y+c_{i}}{T}
\end{equation}
depending on an additional parameter $T\in\mathbb{R}_{>0}$ (called \emph{smoothness}) and defined on the subset of $\mathbb{R}^2$ where the sum is positive.
A \emph{smoothed patchwork} is the curve in the plane 
given by $F(x,y)=0$.
\end{definition}

The smaller the smoothness parameter $T$ is, the closer the smoothed patchwork is to 
a patchwork with the same remaining parameters. \journalvsarxiv{See Theorem~A.1
in the supplementary material.}{See Theorem~\ref{th-tropical-limit} in Appendix~\ref{sec:dependence}.} 
We obtain some patchwork complexity bounds in Appendix~\ref{sec:representation-complexity}, \journalvsarxiv{}{ (see Tab.~\ref{tab-bounds})}.

\subsection{3D}
The construction in 3D is completely analogous. We represent a closed surface $S$ as the zero level set of a field $f\colon\mathbb{R}^3 \to \mathbb{R}$
of a certain special form. Compared to a polygonal mesh, this gives a simple 
way to determine if a point $x$ is inside $S$ by computing the sign of $f(x)$.

\begin{definition}[See Fig.~\ref{fig:sphere}] 
A \emph{patchwork field} is a function $\mathbb{R}^3\to\mathbb{R}$ of the form
\begin{equation}
\label{eq-def-3d-patchwork-limiting}
f(x)=\max_{\substack{1\le i\le n:\\s_i=+1}}\left\{\langle a_i, x \rangle+c_{i}\right\}-\max_{\substack{1\le i\le n:\\s_i=-1}}\left\{\langle a_i, x \rangle+c_{i}\right\},
\end{equation}
depending on the parameters $n\in\mathbb{Z}_{>0}$ (called \emph{width}), 
$s_i\in \{+1,-1\}$ (
\emph{signs}), $a_i \in \mathbb{R}^3$, and $c_i\in \mathbb{R}$, where the first maximum is taken over indices $i\in\{1,\dots,n\}$ such that $s_i=+1$ and the second maximum is over $i$ such that $s_i=-1$. 
A \emph{patchwork} is the surface in $\mathbb{R}^3$ 
given by $f(x)=0$.

We also introduce a \emph{smoothed patchwork field}
\begin{equation}
\label{eq-def-3d-patchwork}
F(x)=T
\log\sum_{i=1}^n s_i\exp \frac{\langle a_i, x \rangle +c_{i}}{T}
\end{equation}
depending on an additional parameter $T\in\mathbb{R}_{>0}$ (called \emph{smoothness}) and defined on the subset of $\mathbb{R}^3$ where the sum is positive. 
A \emph{smoothed patchwork} is the surface in $\mathbb{R}^3$ 
given by $F(x)=0$.
\end{definition}

\begin{figure}
    \centering
    \begin{subfigure}{0.15\textwidth}
        \includegraphics[width=\linewidth]{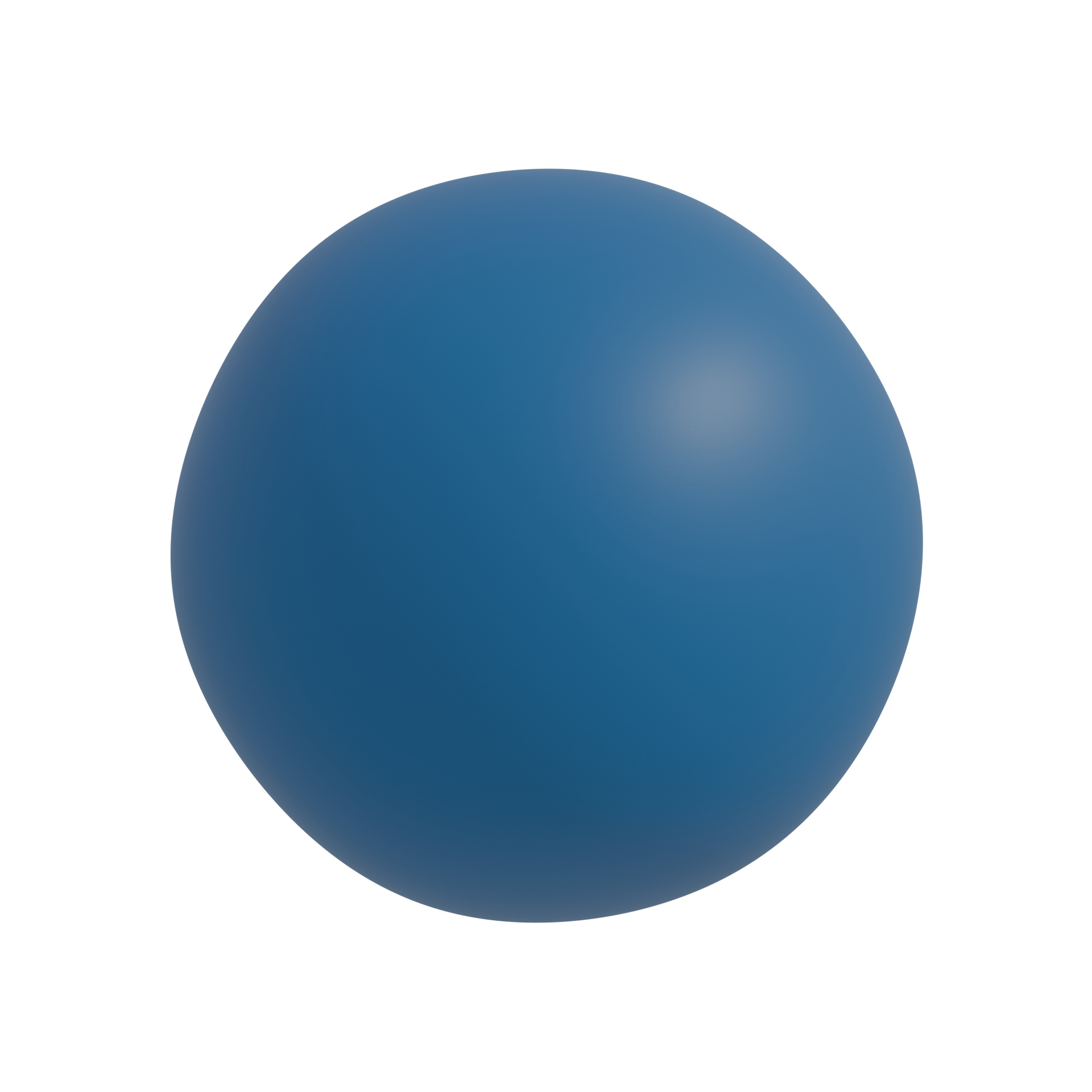}
        \caption{}
    \end{subfigure}
    \begin{subfigure}{0.15\textwidth}
        \includegraphics[width=\linewidth]{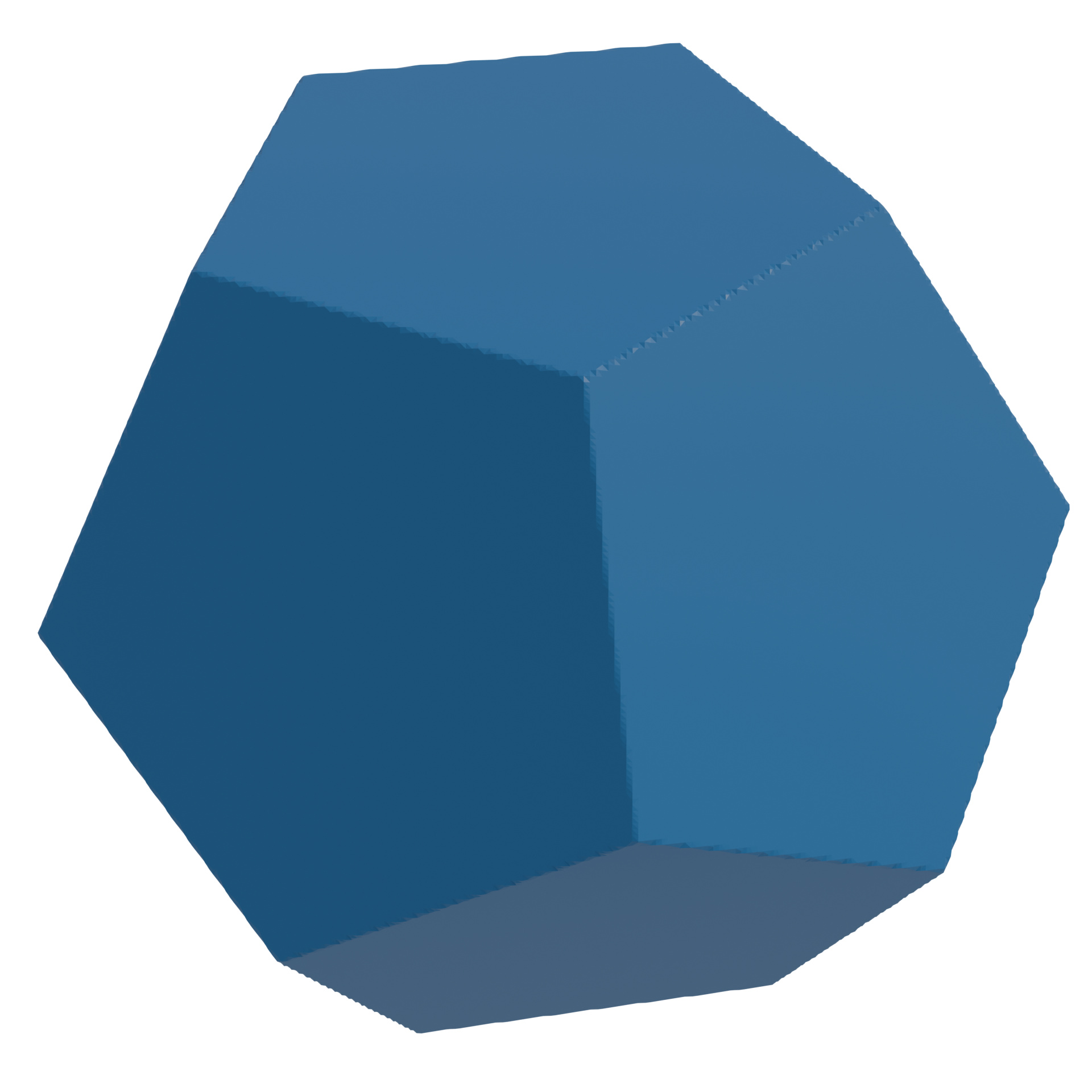}
        \caption{}
    \end{subfigure}
    \hfill\hfill
    \begin{subfigure}{0.15\textwidth}
        \includegraphics[width=\linewidth]{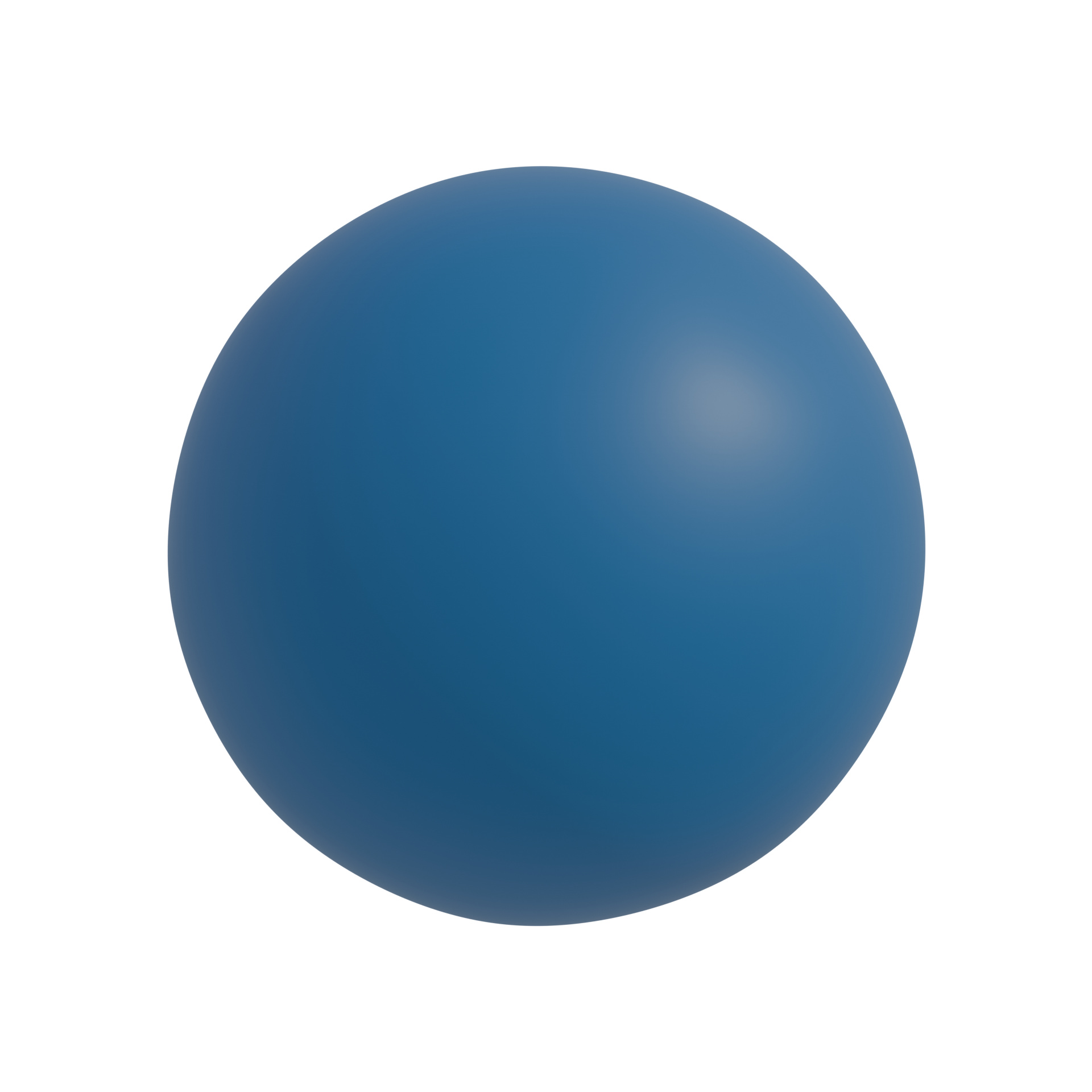}
        \caption{}
    \end{subfigure}
    \begin{subfigure}{0.15\textwidth}
        \includegraphics[width=\linewidth]{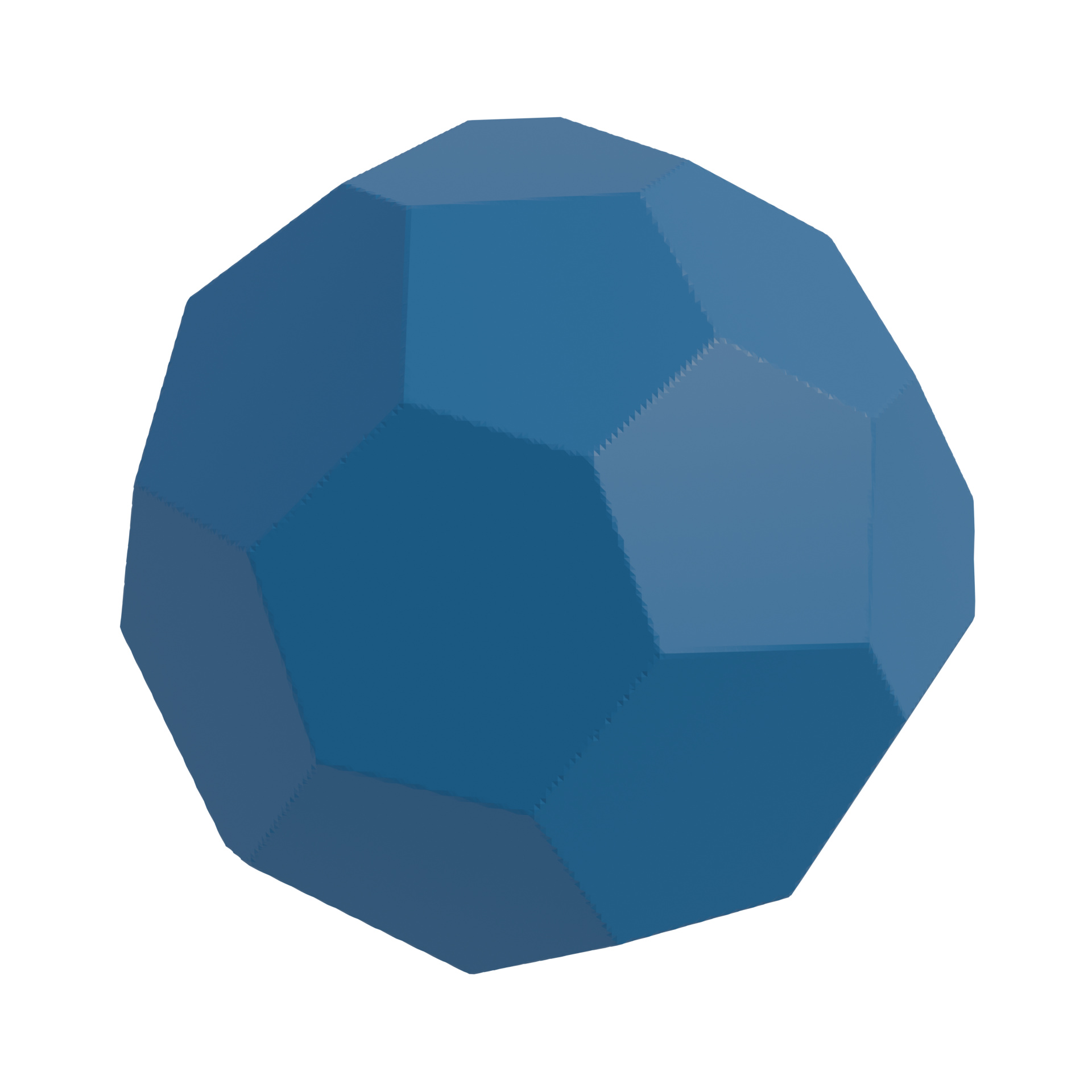}
        \caption{}
    \end{subfigure}
    \hfill\hfill
    \begin{subfigure}{0.15\textwidth}
        \includegraphics[width=\linewidth]{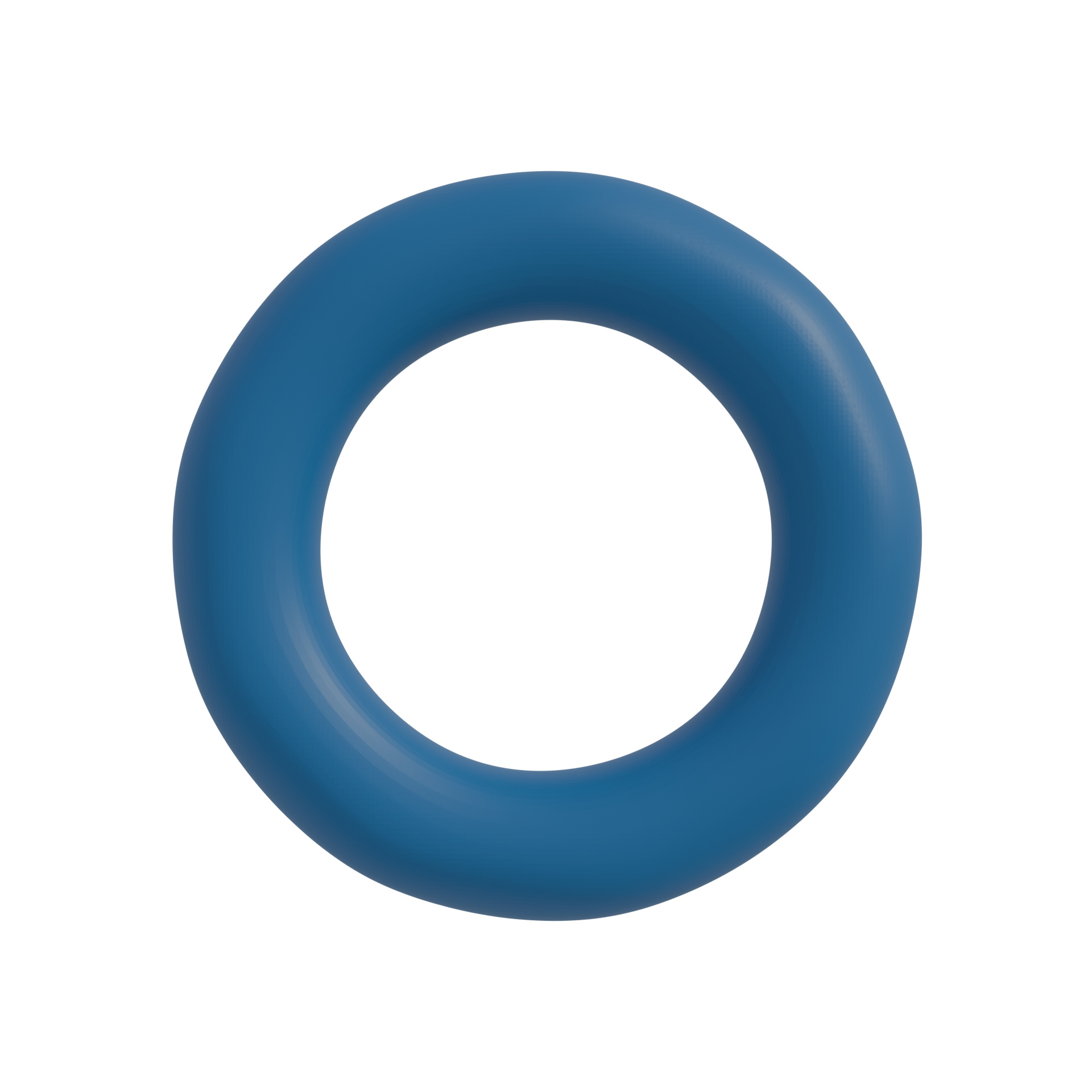}
        \caption{}
    \end{subfigure}
    \begin{subfigure}{0.15\textwidth}
        \includegraphics[width=\linewidth]{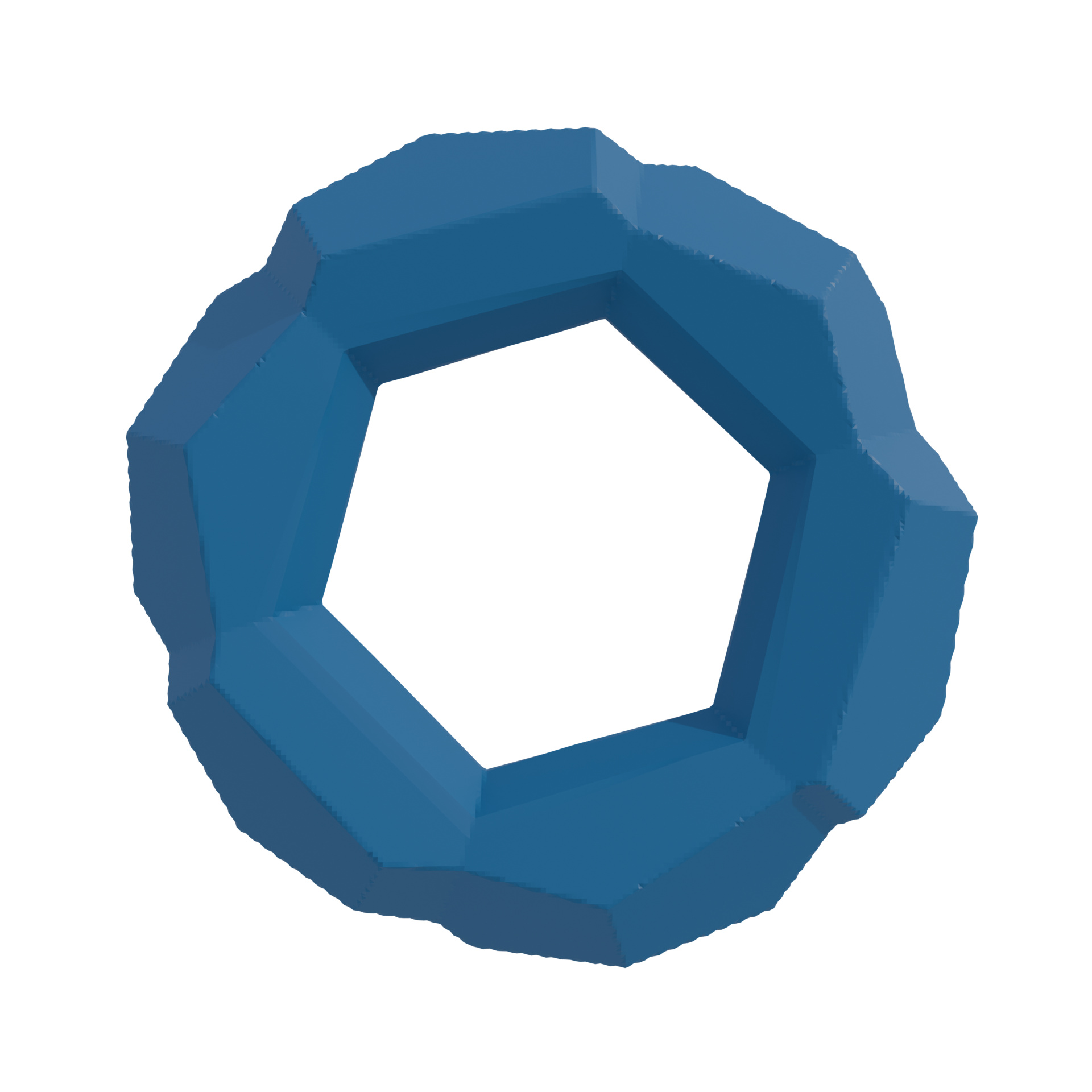}
        \caption{}
    \end{subfigure}
    \caption{Smoothed patchworks approximating a sphere (a and c) and a torus (e) and their limiting cases 
    $T\searrow 0$ (b, d, f). The smoothed patchworks (a, c, e) 
    visually coincide with the shapes, although their width is just $n=13$, $31$, $30$ (left to right). For (a and b), the $12$ design planes contain the faces of the dodecahedron (b), and the $13$th linear function is constant. For (e and f), just $30$ design planes
    represent an intricate nonconvex shape. 
    }
    \label{fig:sphere}
\end{figure}

\textbf{Visualization.} 
Introduce linear functions $\mathbb{R}^3\to\mathbb{R}$ 
\begin{equation}\label{eq-design-planes}
    l_i(x)=\langle a_i, x \rangle + c_i. 
\end{equation}
The planes $l_i(x)=0$ (with $l_i(x)\ne\mathrm{const}$) are called \emph{design planes}. The planes $l_i(x)=l_j(x)$ for distinct $i,j\in\{1,\dots,n\}$ are called \emph{equality planes}. The (possibly unbounded) convex polygons given by the conditions 
\begin{equation}
l_i(x)=l_j(x)\ge \max_{k}l_k(x)
\end{equation}
are called \emph{candidate polygons}. \emph{Active polygons} are the candidate polygons such that $s_is_j<0$. The candidate polygons decompose space into convex polyhedra (possibly empty or unbounded) given by the condition $l_i(x)\ge \max_{k}l_k(x)$ for some~$i$. These polyhedra are called \emph{candidate polyhedra}.

For 
general position functions $l_i(x)$ 
including a positive constant, 
the patchwork is just 
the union of active polygons. For 
small smoothness $T>0$, the smoothed patchwork is close to the union. \journalvsarxiv{See Theorem~A.2 
in the supplementary materials.}{See Theorem~\ref{th-tropical-limit-3d} 
in Appendix~\ref{sec:dependence}.} This leads to compact surface representations; 
see Fig.~\ref{fig:sphere}.
Expressive power and complexity bounds of 3D (smoothed) patchworks are established in \journalvsarxiv{Theorem~B.2 
and Table~7 in the supplementary materials.}{Theorem~\ref{th-3d-representation-property} in Appendix~\ref{sec:proofs} and Tab.~\ref{tab-3d-bounds} in Appendix~\ref{sec:representation-complexity}, respectively.}

\section{Dependence on the parameters}\label{sec:dependence}

In this supplementary material, we discuss the dependence of the patchwork on the parameters and clarify the technical assumptions required for Section~\ref{sec:method} above. The derivations there are valid in general position, that is, unless the parameters satisfy certain specific relations. We show that a general position patchwork $f(x,y)=0$ is a limiting case of $F(x,y)=0$ as 
$T\searrow 0$; see Fig.~\ref{fig-patchwork}.

More precisely, with notation~(\ref{eq-def-patchwork-limiting})--(\ref{eq-design-lines}), we say that $a_i,b_i,c_i
\in \mathbb{R}$ are \emph{general position parameters including a positive constant}, if the following conditions hold:
\begin{itemize}
    \item for some index $i$, we have $l_i(x,y)=\mathrm{const}>0$; 
    \item for distinct indices $i$, the 
    functions $l_i(x,y)$ are distinct;
    \item for any triple of distinct indices $i,j,k$, the equation $l_i(x,y)=l_j(x,y)=l_k(x,y)$ has finitely many solutions;
    \item for any four distinct indices $i_1,\dots,i_4$, the equation $l_{i_1}(x,y)=\dots=l_{i_4}(x,y)$ has no solutions.
\end{itemize}
By the argument in Appendix~\ref{sec:method}, for general position parameters including a positive constant, the patchwork $f(x,y)=0$ is the union of active segments. Here, none of the 
conditions 
can be dropped. For instance, for $l_1(x,y)=x$, $l_2(x,y)=y$, $l_3(x,y)=1$, $l_4(x,y)=(x+y+1)/3$, $s_1=s_2=s_3=-s_4=1$ (violating the last condition), 
the set $f(x,y)=0$ consists of a single point $(1,1)$, whereas 
there are no active segments.
The following folklore theorem holds.

\begin{theorem}[Tropical limit] \label{th-tropical-limit}
With notation~(\ref{eq-def-patchwork-limiting}), (\ref{eq-design-lines}), and~(\ref{eq-def-patchwork}), 
for any general position parameters including a positive constant,
the Hausdorff distance between the sets given by $F(x,y)=0$ and $f(x,y)=0$ tends to zero as $T\searrow 0$. 
\end{theorem}

\begin{proof}
    Set $\beta:=1/T$. Take $\varepsilon>0$ and 
    consider the $\varepsilon$-neighborhood $U$ of the patchwork $f(x,y)=0$ 
    and the candidate polygon $P_i$ corresponding to an index $i$. 
    It suffices to show that the sign of $e^{\beta F(x,y)}-1$ 
    in $P_i\setminus U$ 
    equals $s_i$ for sufficiently large $\beta$, where the function $e^{\beta F(x,y)}$ is extended to points $(x,y)$ where the argument of the logarithm in Eq.~(\ref{eq-def-patchwork}) is negative.
    Clearly, then the Hausdorff distance between the sets given by 
    $F(x,y)=0$ and $f(x,y)=0$ is $O(\varepsilon)$ and tends to zero as $T\searrow 0$. 

    We start by showing that 
    $f(x,y)$ is bounded from $0$ in $\mathbb{R}^2\setminus U$. 
    Indeed, take a box $B$ containing all pairwise intersections of equality lines. The complement of $B\cup U$ can be covered by rays starting at the boundary of $B$ and not crossing the equality lines. On each ray, 
    $f(x,y)$ is linear and non-vanishing. Hence, $|f(x,y)|$ attains a minimum at the starting point. By compactness, there is $\delta=\delta(\varepsilon)>0$ such that $|f(x,y)|>\delta$ in $B\setminus U$, hence in $\mathbb{R}^2\setminus U$. 

    Next,
    take $(x,y)\in P_i\setminus U$. By the definition of $P_i$, the sign of $f(x,y)$ equals 
    $s_i$. Assume that $s_i=+1$ (the other case is analogous). Then $f(x,y)>\delta$. 
    Pick any 
    \begin{equation*}
        j\in\arg\max_{k:s_k<0}\{l_k(x,y)\}.
    \end{equation*}
    Then $l_i(x,y)>\delta+l_j(x,y)$. 
    Since some $l_k(x,y)=c>0$ is a positive constant, it follows that $l_i(x,y)\ge c$.  
    Then 
     \begin{equation}
     \begin{aligned}
         &e^{\beta F(x,y)}=\sum_{k:s_k>0}
         e^{\beta l_k(x,y)}-\sum_{k:s_k<0}
         e^{\beta l_k(x,y)}
         \ge 
         \begin{cases}
         \left(
         e^{\beta \delta}-n
         \right)e^{\beta l_j(x,y)},
         &\text{if }l_j(x,y)\ge 0,\\
         e^{\beta c}-n,
         &\text{if }l_j(x,y)< 0,
         \end{cases}
         \end{aligned}
    \end{equation}
    We conclude that $e^{\beta F(x,y)}>1$, once
    \begin{equation}\label{eq-beta-bound}
        \beta>\frac{\log(n+1)}{\min\{\delta,c\}}.
    \end{equation}
    Analogously, $e^{\beta F(x,y)}<1$ in $P_i\setminus U$ 
    for $s_i=-1$. QED 
\end{proof}

Let us give a few remarks on the dependence on the parameters other than $T$. 
The change of $a_i,b_i,c_i,s_i$ affects only the boundary of the candidate polygon corresponding to index $i$, and hence is a local change of the patchwork $f(x,y)=0$. The changes to the smoothed patchwork $F(x,y)=0$ exponentially decrease with distance from the polygon. When a sign $s_i$ flips, 
a polygon is abruptly added to or removed from the ``interior'' of the patchwork $f(x,y)=0$. For the smoothed patchwork $F(x,y)=0$, a similar jump occurs.

In 3D, we say that $a_i\!\in\! \mathbb{R}^3$\,and\,$c_i\!
\in\! \mathbb{R}$ are \emph{general position parameters including a positive constant} if 
\begin{itemize}
    \item for some index $i$, we have $l_i(x)=\mathrm{const}>0$;
    \item for any $k\ge 2$ distinct indices $i_1,\dots,i_k$, we have 
    \begin{equation}
        \dim\{x\in\mathbb{R}^3:l_{i_1}(x)=\dots=l_{i_k}(x)\}\le 4-k.
    \end{equation}
\end{itemize}
For such parameters, the patchwork $f(x)=0$ is the union of active polygons. 
The following folklore theorem is proved analogously to Theorem~\ref{th-tropical-limit}.

\begin{theorem}[Tropical limit] \label{th-tropical-limit-3d}
With notation~(\ref{eq-def-3d-patchwork-limiting})--(\ref{eq-design-planes}), 
for any general position parameters including a positive constant, 
the Hausdorff distance between the sets given by 
$F(x)=0$ and $f(x)=0$ tends to zero 
as $T\searrow 0$. 
\end{theorem}

\section{Expressive power}\label{sec:proofs}

In this supplementary material, we demonstrate the expressive power of patchworks: first, in 2D, then, in 3D. We prove the following folklore theorem.

\begin{theorem}[Expressive power] \label{th-representation-property-2D}
    For any Jordan curve and any $\varepsilon>0$, there exists a patchwork (that is again a Jordan curve) within the Hausdorff 
    distance $\varepsilon$ from the curve. For a smooth curve, there is a patchwork of width $O(1/\varepsilon^2)$ with this property.
\end{theorem}

Here, roughly speaking, a Jordan curve is a closed, continuous curve without self-intersections. More precisely, a \emph{Jordan curve} is the image of an injective continuous map of a circle into the plane. A Jordan curve need \emph{not} be smooth, so that the theorem applies to fractal curves as well. The generalization to curves having finitely many disjoint components is straightforward. Thus, we obtain a method to approximate essentially any plane curve by a patchwork.

Recall that \emph{any} 
piecewise linear function can be represented in form~(\ref{eq-def-patchwork-limiting}). 
Thus, 
\emph{any} polyline without self-intersections or a disjoint union of such polylines can be represented as a patchwork. The advantage of our 
method is a simple, explicit expression for the parameters.

\begin{proof}
We prefer to use a hexagonal lattice for approximation. Replace the index $i$ 
in~(\ref{eq-def-patchwork}) with a pair of indices $k,l\in [-N,N]$, so we write $a_{k,l}$ instead of $a_{i}$ etc. Consider a patchwork with parameters 
\begin{equation}\label{eq-hex}
a_{k,l}=k, \quad 
b_{k,l}=l, \quad 
c_{k,l}=-(k^2+l^2+kl-1)/N, \quad
s_{k,l}\in\{+1,-1\}.
\end{equation}
The resulting candidate polygons are hexagons with the sides parallel to the lines $x=0$, $y=0$, and $x=y$. 
They form an affine-transformed honeycomb lattice; see Fig.~\ref{fig:hex} on the right. 
The resulting patchwork (that is, the union of polylines separating candidate polygons corresponding to the indices $k,l$ with $s_{k,l}$ of opposite signs) is called a \emph{hex-curve} (\emph{of step $1/N$}). 

It suffices to prove Theorem~\ref{th-representation-property-2D} for a smooth curve, because any Jordan curve can be approximated by smooth ones in Hausdorff distance. 
By \citep[Lemma~3.8]{KSS-25}, for any smooth Jordan curve, there is a constant $C$ (depending on the curve) such that for any sufficiently small mesh step $\delta>0$ there is a Jordan hex-curve of step $\delta$ within Hausdorff distance $C\delta$ from the smooth curve. For $\delta<\varepsilon/C$, the hex-curve is within Hausdorff distance $\varepsilon$. Assume WLOG that the smooth curve lies in the box $|x|,|y|<1$. Then the hex-curve is the patchwork with parameters~(\ref{eq-hex}) and has width at most $(2/\delta)^2=O(1/\varepsilon^2)$, as required. QED
\end{proof}

Note that for a non-smooth curve, the Hausdorff distance to the hex-curve can decrease with $\delta\to 0$ slower than any linear function in $\delta$ \citep[Remark after Lemma~3.8]{KSS-25}.

As a byproduct, our representation allows direct sampling of certain complicated random fractals. Namely, if one paints the hexagons of a honeycomb lattice in two colors randomly with equal probabilities independently of each other, then the hex-curve approximates a random fractal known as CLE$_6$ \citep{Miller-etal-17}. 
Patchwork parameters~(\ref{eq-hex}) can be easily modified to represent the hex-curve (before the affine transformation of the honeycomb lattice) and thus to sample the random CLE$_6$ fractal. Other random fractals arising from sampling random patchworks can also be of interest and are the subject of future work.

Another side remark is that the relation between the patchwork parameters and the tesselation by candidate polygons is most visual in the context of \emph{isotropic geometry} and \emph{metric duality}; see \citep[Section~2]{pirahmad2025} for an introduction. 

Now we proceed to the 3D case.
Let us approximate a given surface $S\subset [-1,1]^3$ by a patchwork. Take an integer $N$ and 
replace the index $i$ 
in~(\ref{eq-def-3d-patchwork}) with a triple of indices $k,l,m\in [-N,N]$, so we write $a_{k,l,m}$ instead of $a_{i}$ etc. Set 
\begin{equation}\label{eq-3d-grid} 
a_{k,l,m}=(k,l,m), \!\quad c_{k,l,m}= \frac{1-k^2-l^2-m^2}{2N}, 
\!\quad s_{k,l,m}=\begin{cases}
-1, &\text{if $\frac{(k,l,m)}{N}$ is inside $S$,}\\ 
+1, &\text{otherwise.}
\end{cases}
\end{equation}
The resulting candidate polygons are squares forming a cubic grid. 
The patchwork is the union of the squares separating cube midpoints inside and outside~$S$ (despite the parameters not being in general position). 
This union is called a \emph{digital surface}. 
%
%
We arrive at the following folklore theorem.

\begin{theorem}[Expressive power] \label{th-3d-representation-property}
    Any compact embedded surface without boundary in $\mathbb{R}^3$ can be approximated by patchworks, that is, for each $\varepsilon>0$, there exists a patchwork within the Hausdorff distance $\varepsilon$ from the surface.
\end{theorem}

Here, by a \emph{compact embedded surface without boundary}, we mean the image of an injective continuous map of a compact two-dimensional 
manifold into $\mathbb{R}^3$. The manifold need not be connected, and the map need not be smooth.

\begin{proof}
It suffices to show that the Hausdorff distance between the initial surface $S$ and the digital surface $D$ tends to zero as $\delta\to 0$, because any digital surface is a patchwork with parameters~(\ref{eq-3d-grid}). Consider the set $S_\delta$ of points inside $S$ at a distance $\ge 2\delta$ from it. All those points lie in the cubes inside $S$; otherwise, there would be a cube of diameter $\sqrt{3}\,\delta$ intersecting both  $S$ and $S_\delta$. All the points outside of $S$ lie in the cubes, not inside $S$. Thus, $D$ lies between $S$ and $S_\delta$, and the Hausdorff distance between $S$ and $D$ is not greater than the one between $S$ and $S_\delta$. The latter distance tends to zero as $\delta\to 0$: indeed, applying Dini's theorem to the 
functions $f_\delta\colon S\to\mathbb{R}$, $x\mapsto\mathrm{dist}(x,S_\delta)$, we conclude that $\max_{x\in S} f_\delta(x)\to 0$ as $\delta\to 0$. QED
\end{proof}

\section{Representation complexity}
\label{sec:representation-complexity}

Now we address the question of the \emph{representation complexity} of patchworks, that is, what complexity of patchworks
is required to represent the
shape of given complexity. Remarkably, for some complexity measures, we can achieve lower bounds (cf.~\emph{problem of lower bounds} in computational complexity theory).

In this work, we focus on a single toy example of a complexity measure to provide a proof of concept for patchwork complexity lower bounds. Namely, we study rough bounds on the numerical precision (the number of digits in the decimal representation) of the patchwork coefficients $a_{i}$ and $b_i$ ($a_{i}$ in the 3D case) required to represent a given shape.  
So far, we have not attempted to bound the width $n$, which is a more natural but more complicated complexity measure for the analysis, to be addressed in future work. 

Without loss of generality, we assume that all the coefficients $a_{i}$ and $b_i$ are integers, because we can always scale them by a common factor and absorb it into $T$, without changing the smoothed patchwork. We also assume that they 
are non-negative, because we can always add a common constant to them and add two auxiliary design lines preserving the smoothed patchwork. Then, introduce the \emph{patchwork precision complexity}
\begin{equation}
\begin{aligned}
M:=\max_{i}\{a_{i}+b_i\}.
\end{aligned}
\end{equation}
The value $\log_{10}M$ roughly measures the numerical precision of the coefficients before scaling and adding a constant.

Analogously, in 3D, denote $a_i=(a_{i,1},a_{i,2},a_{i,3})$, assume 
$a_{i,j}\in\mathbb{Z}_{\ge0}$ for all $i,j$, and set
\begin{equation}
\begin{aligned}
M:=\max_{i}\{a_{i,1}+a_{i,2}+a_{i,3}\}.
\end{aligned}
\end{equation}

\begin{table*}[t]
    \footnotesize
\centering
\setlength{\aboverulesep}{1pt}
\setlength{\belowrulesep}{1pt}
\caption{A few lower bounds on the precision complexity $M$ of a smoothed patchwork representing a given bounded curve. The \emph{$x$-projection degree} $D$ is the maximal number of intersections of the curve with a $y$-parallel line not contained in the curve.
See Section~\ref{sec:representation-complexity} for details.}
\vspace{2mm}
\begin{tabular}{@{}llll}
     \textbf{Curve complexity measure \quad }&\textbf{Notation \quad }&\textbf{
    Bound} & \textbf{Underlying theorem}   \\\hline \\
    $x$-projection degree
    &
    $D$ &
    $M\ge D$ 
    &
    B\'ezout's theorem\\
    number of components &
    $C$ &
    $M\ge \sqrt{2C}$ \quad 
    &
    Harnack's curve theorem
    \\
\end{tabular}
\label{tab-bounds}
\end{table*}

\begin{table*}[t]
    \footnotesize
\centering
\setlength{\aboverulesep}{1pt}
\setlength{\belowrulesep}{1pt}
\caption{A few lower bounds on the precision complexity $M$ of a smoothed patchwork representing a given closed surface (having a smooth extension to complex projective space). The \emph{$xy$-projection degree} $D$ is the maximal number of intersections of the surface with a $z$-parallel line not contained in the surface. The \emph{number of handles} $G$ is the total genus of all the surface components. 
}
\vspace{2mm}
\begin{tabular}{@{}llll}
     \textbf{Surface complexity measure \quad
     }& \textbf{
     \quad}&\textbf{
     Bound} & \textbf{Underlying theorem}   \\\hline \\
    $xy$-projection degree
    &
    \hspace{-1cm} $D$ &
    $M\ge D$
    &
    B\'ezout's theorem\\
    number of components \quad &
    \hspace{-1cm} $C$ &
    $M\ge \sqrt[3]{2C}$ 
    &
    Generalized Harnack Inequality 
    \\
    number of handles  &
    \hspace{-1cm} $G$ &
    $M\ge \sqrt[3]{2(C+G)}$ \quad
    &
    \citep{itenberg2006asymptotically} and Proposition~\ref{prop-bound} 
    \\
\end{tabular}
\label{tab-3d-bounds}
\end{table*}

In Tables~\ref{tab-bounds}--\ref{tab-3d-bounds}, we collect a few lower bounds on $M$ for a smoothed patchwork with $T>0$ 
and all $a_{i},b_i\in\mathbb{Z}_{\ge0}$ (or $a_{i,j}\in\mathbb{Z}_{\ge0}$) 
representing a given shape. They come from the observation that the \emph{exponential map} $(x,y)\mapsto \left(e^{
x/T},e^{
y/T}\right)$ 
takes the smoothed patchwork to an algebraic set of degree~$M$, and the known theorems on the complexity of the latter.

Conversely, if for some $
T>0$, the exponential map
takes a given set to an algebraic hypersurface of degree~$M$ (in the positive orthant), then the set is a smoothed patchwork of precision complexity $M$. This is a toy example of a \emph{patchwork complexity upper bound}.

Specifically, in Tables~\ref{tab-bounds}--\ref{tab-3d-bounds}, we use B\'ezout's and Harnack's theorems, and the following result. 

\begin{proposition}[Algebraic surface complexity bound] \label{prop-bound} Consider a bounded algebraic surface of degree $M$ in $\mathbb{R}^3$ with $C$ components and the sum $G$ of the genera of all components. Assume that the extension of the algebraic surface to complex projective space is smooth. 
Then 
\begin{equation}
M(M^2-4M+6)\ge 2(C+G).
\end{equation}
\end{proposition}

Proposition~\ref{prop-bound} follows from known results surveyed in~\citep{itenberg2006asymptotically}, but the proof requires familiarity with basic topology and algebraic geometry.   

\begin{proof}
Since the algebraic surface is bounded and non-singular, it is a disjoint union of finitely many closed compact orientable surfaces, and its projective closure in real projective space 
has the same topological type. Then the projective closure has
the sum $C+2G+C$ of the Betti numbers over $\mathbb{Z}/2\mathbb{Z}$. 
By the Generalized Harnack Inequality \citep[Theorem~1.1]{itenberg2006asymptotically}, this sum is not greater than the analogous sum for the extension (Zariski closure) to the complex projective space.
By the known computation (see e.g., \citep{Yuan-14} or \citep[p.~93]{itenberg2006asymptotically}), the latter sum equals $M(M^2-4M+6)$, and the desired inequality follows.  QED
\end{proof}

\section{Future theoretical work}
\label{sec:theoretical-future-work}
Patchworks represent piecewise-linear shapes, and their smoothed versions are smooth. It is interesting to modify the construction to allow sharp corners and edges for non-piecewise-linear curves and surfaces.
It is interesting to get bounds on the precision complexity and width of the patchwork representation of a shape; for instance, to determine if the maximal number of intersections of the shape with a line is a lower bound on both. 
Furthermore, we conjecture that the upper bound $O(1/\varepsilon^2)$ on the width described in Theorem~\ref{th-representation-property-2D} 
can be improved to $O(1/\varepsilon)$.

\begin{figure}[h]
    \centering
    \includegraphics[width=1\linewidth]{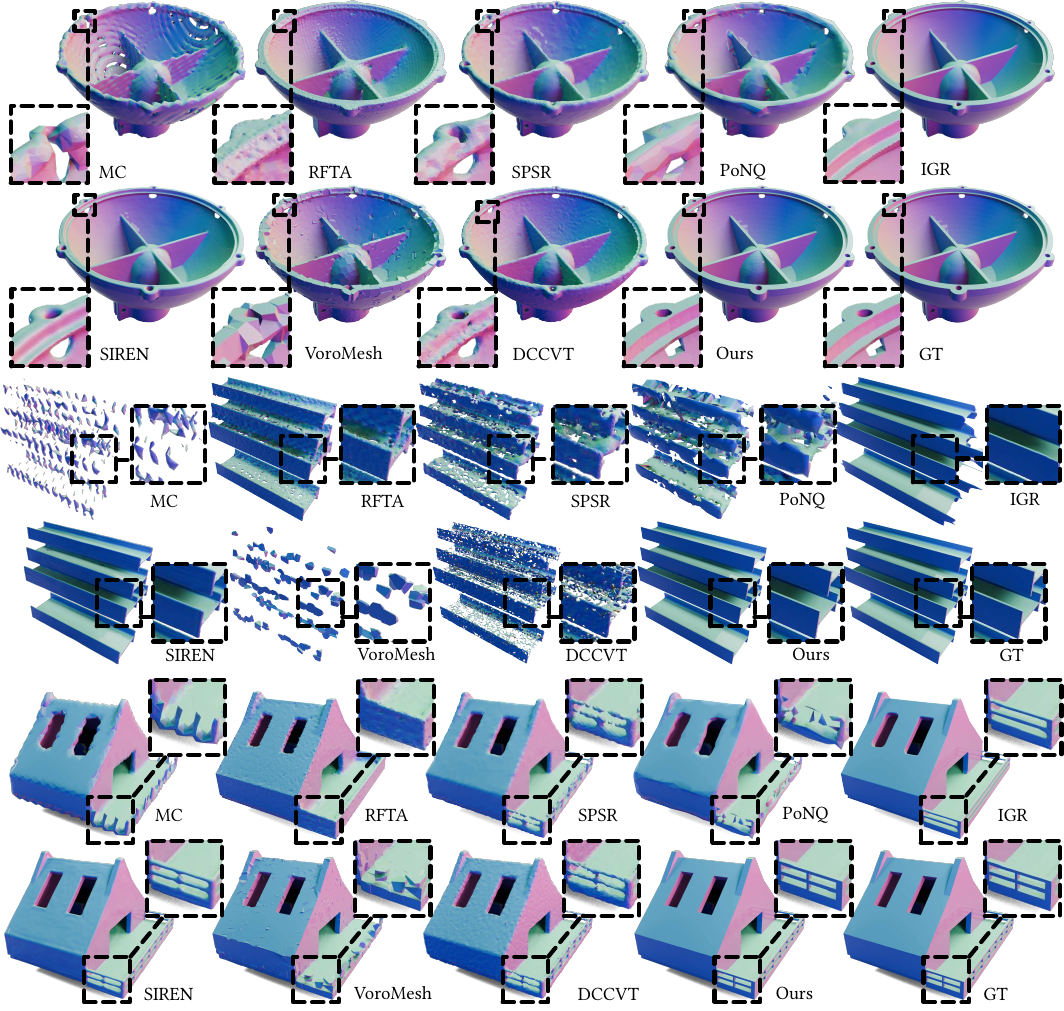}
    \caption{Extended visual results on ABC~\citep{koch2019abc}.}
    \label{fig:abc-ext}
\end{figure}

\begin{figure}[h]
    \centering
    \includegraphics[width=1\linewidth]{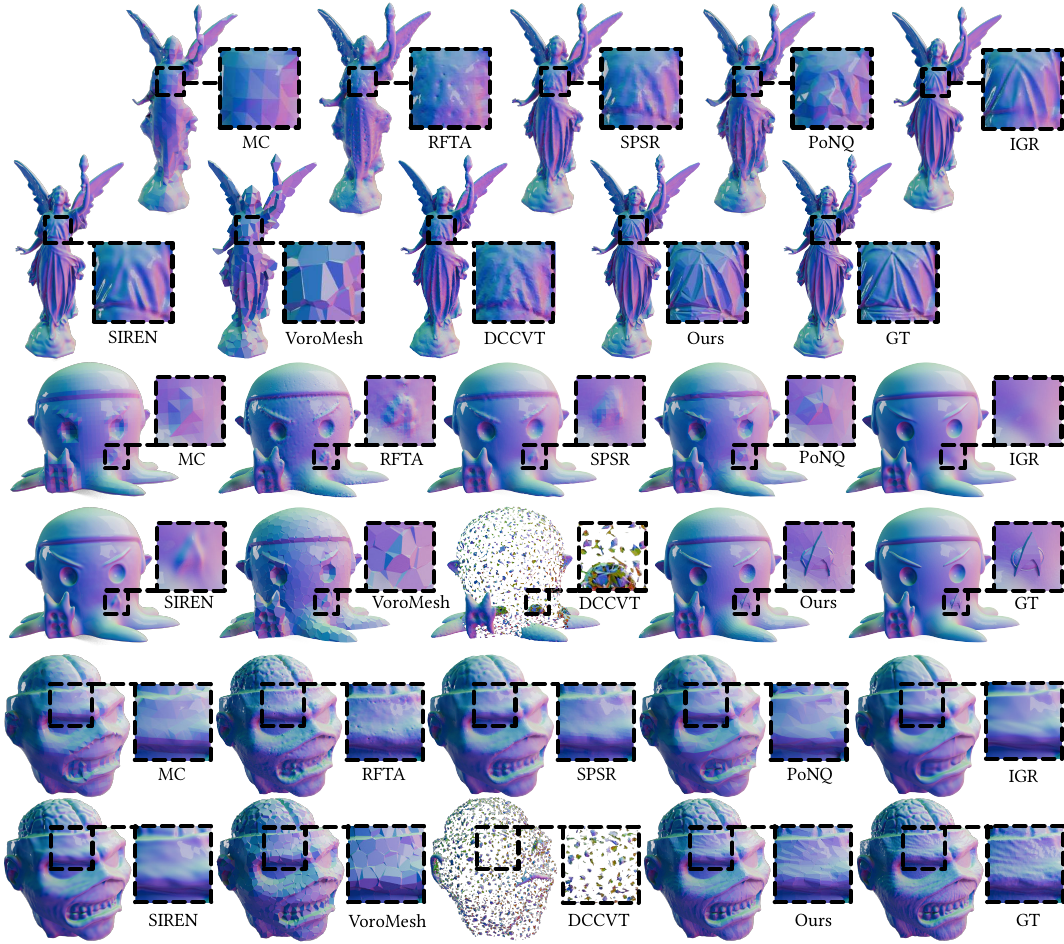}
    \caption{Extended visual results on Thingi10k~\citep{zhou2016thingi10k}.}
    \label{fig:thingi10k-ext}
\end{figure}

\begin{figure}[h]
    \centering
    \includegraphics[width=1\linewidth]{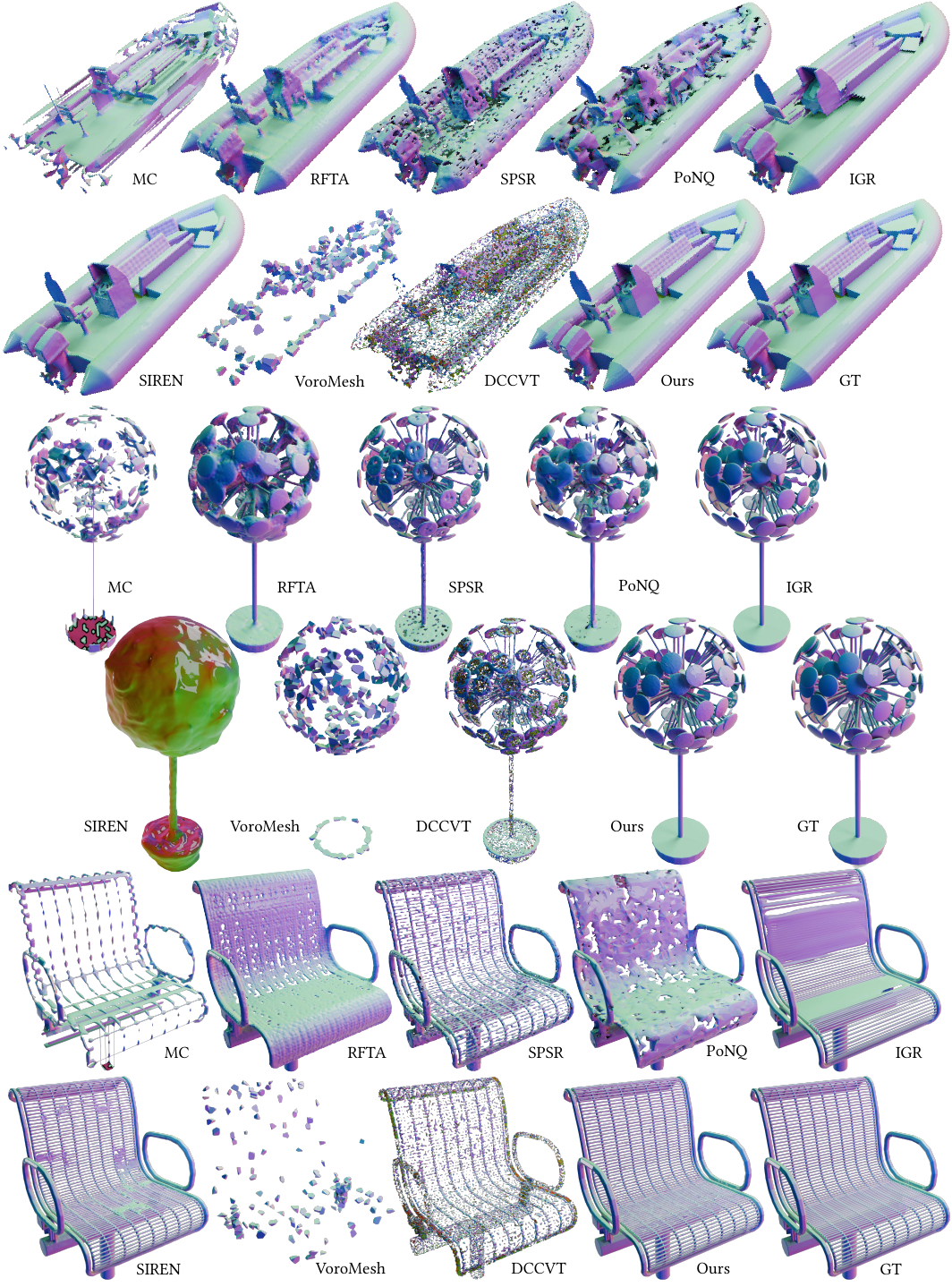}
    \caption{Extended qualitative results on ShapeNet~\citep{chang2015shapenetinformationrich3dmodel}.}
    \label{fig:shapenet-ext}
\end{figure}

\section{Experiment details}
\subsection{Additional visual results}
\label{sec:extended-visual-results}

We provide extended visual comparisons for ABC (Fig.~\ref{fig:abc-ext}), Thingi10k (Fig.~\ref{fig:thingi10k-ext}) and ShapeNet (Fig.~\ref{fig:shapenet-ext}) datasets.

\subsection{Baseline configurations}
\label{sec:baseline-details}

The configurations of baselines and ours are summarized in Tab.~\ref{tab:baseline-config}.

\begin{table}[t]
  \caption{Baseline configurations at matched parameter counts for Roof Modeling (Tab.~\ref{tab:roof-modeling}) / ABC$\&$Thingi10k (Tab.~\ref{tab:abc_thingi10k}) / Shapenet (Tab.~\ref{tab:shapenet}) datasets.}
  \label{tab:baseline-config}
  \vspace{-10pt}
  \begin{center}
  \resizebox{\linewidth}{!}{
  \begin{tabular}{lccccccc}
    \bf Method & \bf \#Prims & \bf \#Params/Prims & \bf \#Layers & \bf \#Neurons & \bf \#Params \\ \hline \\
    MC & $20^3$ / $51^3$ / $81^3$ & 1 & & & 8000 / 132651 / 531441 \\
    RFTA & $20^3$ / $51^3$ / $81^3$ & 1 & & & 8000 / 132651 / 531441 \\
    SPSR & 1365 / 21845 / 87381 & 6 & & & 8190 / 131070 / 524286 \\
    PoNQ & 512 / 8192 / 32768 & 16 & & & 8192 / 131072 / 524288 \\
    IGR & & &	4 &	64 / 256 / 512 & 8641 / 132094 / 526334 \\
    SIREN & & & 4 & 64 / 256 / 512 & 8641 / 132865 / 527873 \\
    VoroMesh & $14^3$ / $35^3$ / $56^3$ & 3 & & & 8232 / 128625 / 526848 \\
    DCCVT & $14^3$ / $35^3$ / $56^3$ & 3 & & & 8232 / 128625 / 526848 \\
    Ours Init & 2048 / 32768 / 131072 & 4 & & & 8192 / 131072 / 524288 \\
  \end{tabular}}
  \end{center}
  \vspace{-10pt}
\end{table}

\textbf{Regular grid}. Given an SDF grid, the de facto way to extract the underlying surface is Marching Cubes (MC)~\citep{lorensen1987mc}, but it overlooks the SDF's spherical property \citep{reachspheres}, underutilizing the stored information. Reach for the Arc (RFTA)~\citep{Sellan2024RFTA} explicitly exploits this property and achieves the state-of-the-art performance at recovering high-quality surfaces from discrete grids. We compare MC and RFTA on SDF grids computed from ground-truth mesh data. We count 1 float per grid cell, and round the grid resolution to match the budget.

\textbf{Point set}. Given oriented point clouds, the axiomatic method, Screened Poisson Surface Reconstruction (SPSR)~\citep{kazhdan2013spsr}, fits an implicit indicator function and then extracts its isosurface. PoNQ~\citep{maruani2024ponq} associates each oriented sample with a local $\mathbb{R}^{4 \times 4}$ quadric matrix to guide triangulation for better sharp-edge and topology handling. We count the number of floats per budget, that is, 6 for SPSR (point and normal in $\mathbb{R}^3$) and 16 for PoNQ (point and normal in $\mathbb{R}^3$, symmetric $4\times4$ matrix stored in 10 floats), then round the point set count to match the budget.

\textbf{Neural implicit}. Neural implicit representations use coordinate MLPs to encode SDFs as continuous functions, hence requiring either an (approximately) piecewise linear activation function or high expressiveness to model discrete jumps such as sharp edges. We compare IGR~\citep{gropp2020implicitgeometricregularizationlearning}, which pairs Softplus activations of large beta with Eikonal regularization to better model piecewise-smooth geometries. We also compare SIREN~\citep{sitzmann2020implicit}, which leverages sine activations to model high-frequency signals. For simplicity, we fix the number of layers to 4, ignore the input layer parameter count, and round widths to conventional choices. Following prior work~\citep{benshabat2023digsdivergenceguided,yang2023steikstabilizingoptimizationneural,wang2023neuralsingularhessianimplicitneuralrepresentation,wang2024hotspotscreenedpoissonequation}, we fit each shape for 10K iterations.

\textbf{Foam-based}. VoroMesh \citep{maruani2023voromesh} fits double covering point clouds as Voronoi generators and extracts crossing boundaries as reconstructed meshes. Since they initialize generators using voxel centers, we count 3 floats per grid cell and then round the grid resolution. DCCVT~\citep{charawi20263dv} prefits an auxiliary SDF to initialize the SDF values at the Voronoi sites, and constructs an explicit Delaunay tetrahedralization at each optimization step to robustly project Voronoi vertices onto the zero level set. We use the official DCCVT implementation. For each shape, we prefit two HotSpot SDFs for initialization: a converged one (10k iterations) and an unconverged one (500 iterations), and report the best metrics of the two initializations. Following their near-sampling configuration, we initialize sites on a voxel grid at $75\%$ of the target voxel resolution (e.g., $24^3$ for $32^3$) and allocate the remaining site budget to near-surface samples. We ignore the auxiliary SDF when counting their parameter budget.

\begin{table}[h]
\centering
\caption{Quantitative results of GeomDist~\citep{Zhang_2025_ICCV} dataset. The bold text indicates the best score.}
\label{tab:geomdist}
\begin{tabular}{l ccc ccc}
& \multicolumn{3}{c}{\bf Chamfer Distance ($\times 10^{-3}$) $\downarrow$} 
& \multicolumn{3}{c}{\bf \# Parameters} \\ 
\cmidrule(lr){2-4} \cmidrule(lr){5-7}
 & GeomDist & Ours (5M) & Ours (8M) & GeomDist & Ours (5M) & Ours (8M) \\ \hline \\
Archimedes  & 3.301 & 2.168 & \textbf{2.021} & 5532680 & 671144 & 956492 \\
Mouse  & 4.008 & 3.314 & \textbf{2.616} & 5532680 & 780768 & 1037680 \\
Parthenon  & 3.760 & 2.335 & \textbf{2.298} & 5532680 & 864324 & 1079340 \\
Archer  & 3.484 & 2.170 & \textbf{2.150} & 5532680 & 627688 & 925668 \\
Jellyfish  & 4.832 & 2.680 & \textbf{2.629} & 5532680 & 919388 & 1026312 \\
Lamp  & 3.754 & 2.268 & \textbf{2.266} & 5532680 & 992940 & 1143804 \\
Warrior  & 2.774& 1.723 & \textbf{1.718} & 5532680 & 543332 & 807084 \\
Loong  & \textbf{3.706} & 4.440 & 4.425 & 5532680 & 749640 & 937172 \\
Wukong  & \textbf{4.412} & 17.87 & 18.87 & 5532680 & 722716 & 905676 \\
Lion  & \textbf{4.017} & 50.01 & 50.05 & 5532680 & 662820 & 919940 \\
\end{tabular}
\end{table}

\begin{figure}[h]
    \centering
    \includegraphics[width=1\linewidth]{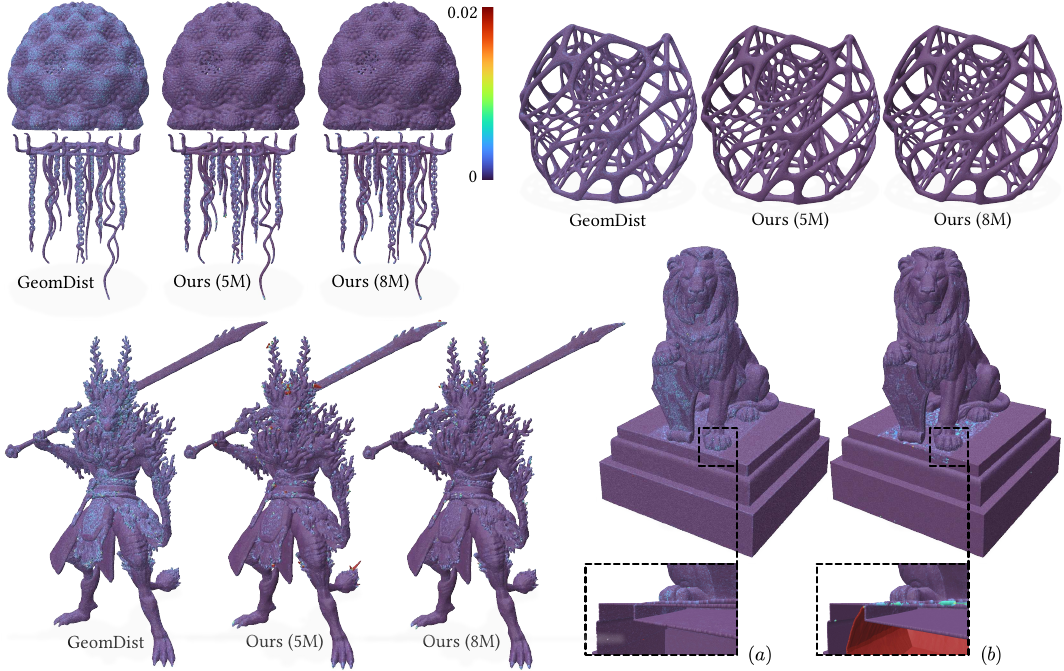}
    \vspace{-10pt}
    \caption{Qualitative results of Jellyfish (top left), Lamp (top right), Loong (bottom left), with 10M samples colored by error map using Polyscope~\citep{polyscope}. As illustrated by the Lion case (Bottom right), although the exterior surface is well represented compared to GeomDist (a), the interior error is large because we cannot handle the topology of self-intersecting structures (b).}
    \label{fig:scale}
\end{figure}

\subsection{Scaling comparison}

We select 10 cases with visually identifiable outer surfaces from the GeomDist~\citep{Zhang_2025_ICCV} shape collection, and fit with $N=655360/1048576$, that is, $5242880/8388608$ parameters at initialization (abbreviated as 5M and 8M). We set batch size $K=65536$, and the scale $\tau=8$. Each $10K$-iteration of fitting takes about $18$/$26$ minutes and uses $644$/$812$ MB of VRAM on a single NVIDIA RTX 3090 GPU with our KDTree variant. For GeomDist, we use the official checkpoints. Following their implementation, we measure Chamfer distance over 10M samples with respect to GT surfaces.

Compared to GeomDist, which requires over 40 hours on 4 A100 GPUs (2.5 min/epoch × 1000 epochs) to fit a single shape, our approach is significantly more efficient and offers additional benefits of an implicit form and explicit polygonalization. As shown in Fig.~\ref{fig:scale} and Tab.~\ref{tab:geomdist}, our model captures faithful global geometry and slightly outperforms GeomDist in mean Chamfer Distance at a matched parameter count for 7 of 10 shapes, with accuracy improving further as parameters increase. However, our formulation cannot handle self-intersecting / non-manifold interior structures, resulting in large Chamfer Distance in the bottom three cases of Tab.~\ref{tab:geomdist}. See Fig.~\ref{fig:scale} (bottom right) for visualization.

To explore extreme scaling, we fit the Wukong case with $N=16777216$ (i.e., over $134$M parameters at initialization), $K=65536$ and $\tau=8$ for $10K$ iterations. It takes 6 hours with 6830MB VRAM to fit on a single NVIDIA RTX 3090 GPU. The final parameter count is $10215668$, and the result is visualized in Fig.~\ref{fig:teaser} (Right).

\subsection{Additional results on Roof Modeling}
To evaluate patchwork's expressiveness on low-complexity geometry, we additionally compare on the Roof modeling~\citep{ren2021intuitiveefficientroofmodeling} dataset, which consists of 3K unique building roof meshes that are geometrically simple yet rich in sharp creases. We randomly select $200$ cases, set $N = 1024$, $K=16384$ and $\tau$ to $0.25$. Our results are slightly better than VoroMesh in metric terms (Tab.~\ref{tab:roof-modeling}), but visually cleaner and much closer to the ground truth (Fig.~\ref{fig:roof-modeling}). Compared to DCCVT, our method preserves sharper creases.

\begin{figure}[h]
    \centering
    \includegraphics[width=1\linewidth]{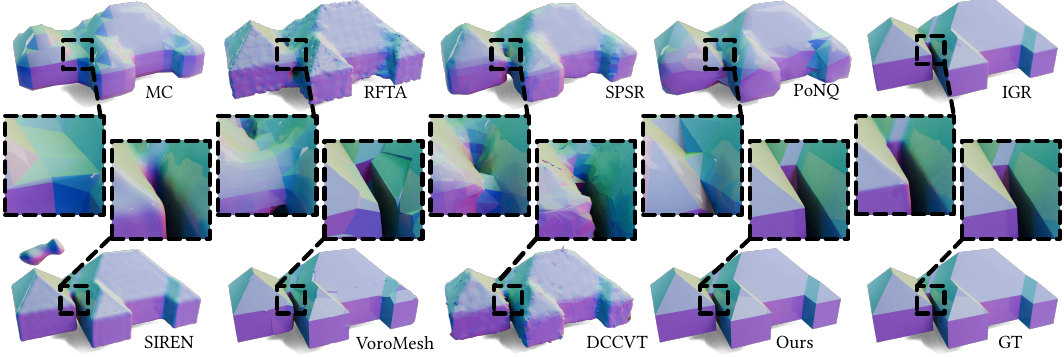}
    \vspace{-10pt}
    \caption{Qualitative results on Roof Modeling~\citep{ren2021intuitiveefficientroofmodeling} dataset.}
    \label{fig:roof-modeling}
\end{figure}

\begin{table}[h]
  \caption{Quantitative results on the Roof Modeling dataset. The bold text indicates the best score; the underlined text indicates the second-best. `Ours \#Params' is the average parameter size after pruning. }
  \label{tab:roof-modeling}
  \vspace{-10pt}
  \begin{center}
  \begin{tabular}{lclccccc}
    \bf Method & \bf CH $\downarrow$ & \bf HD $\downarrow$ & \bf FS $\uparrow$ & \bf NC $\uparrow$ & \bf IoU $\uparrow$ & \bf \#Params \\ \hline \\
    MC & 6.79 & 8.68 & 49.93 & 94.65 & 95.89 & 8000 \\
    RFTA & 3.63 & 5.95 & 36.41 & 96.98 & 97.75 & 8000 \\
    SPSR & 4.57 & 5.88 & 25.86 & 96.95 & 97.21 & 8190 \\
    PoNQ & 3.74 & 9.65 & 62.05 & 96.73 & 98.08 & 8192 \\
    IGR & 1.40 & \underline{1.19} & 77.21 & \underline{99.70} & 99.62 & 8641 \\
    SIREN & 21.2 & 74.0 & 53.79 & 94.05 & 91.02 & 8641 \\
    VoroMesh & \underline{1.29} & 2.05 & \underline{80.45} & 99.44 & \underline{99.92} & 8232 \\
    DCCVT & 3.10 & 5.08 & 45.10 & 96.46 & 96.84 & 8232 \\
    Ours init & 5.90 & 20.0 & 58.65 & 88.03 & 96.84 & 8192 \\
    Ours & \textbf{1.21} & \textbf{0.66} & \textbf{82.16} & \textbf{99.74} & \textbf{99.98} & 2199 \\
  \end{tabular}
  \end{center}
  \vspace{-10pt}
\end{table}

\section{Explicit Polygonalization}
\label{sec:explicit-extraction}

Our method supports extracting the patchwork via explicit polygonalization. For each \emph{interior} candidate polygon/polyderon (i.e., corresponding to an index $i$ with $s_i<0$), 
compute the equality lines/planes 
$\langle a_i,x \rangle + c_i=\langle a_j,x \rangle + c_j$
for all $j=1,\dots,n$
(i.e., the halfspaces that bound the current candidate polygon).
Then find its Chebyshev center $x$ by maximizing $y$ over all pairs $(x,y)$ satisfying the constraint
\begin{equation}
\langle a_{j}-a_i, x \rangle + y ||a_{j}-a_i||_2 + c_{j}-c_{i} \leq 0  
\qquad\text{for each $j=1,\dots,n$}.
\end{equation}

If the resulting $y>0$, then compute the intersections of those halfspaces using convex hull. 
Then we determine connectivity via the secondary convex hull 
computation, and only keep faces shared by other exterior candidate polygons/polydera. After performing the same for all interior candidate polygons/polydera, we merge all active segments/polygons 
for the final patchwork extraction. This extraction procedure is embarrassingly parallel, except for the final merging step.

As mentioned in Sec.~\ref{sec:limitation}, some extruding interior polyhedra may survive the pruning due to erroneous normal information. After explicit polygonalization, they extend far beyond the typical Marching Cubes extraction range, yielding extreme outliers that skew the metric statistics. To address this, we compute the log-diagonal of the bounding box for each interior convex cell and discard those exceeding the third quartile by more than three times the interquartile range~\citep{tukey1977exploratory}.

\subsection{Edge collapsing routine}
\label{sec:edge-collapsing}

During extraction, for each interior convex cell, we compute the area of all polygonal faces, flag those with area less than $10\%$ of the maximal one, and collect their shared candidate polyhedral indices in parallel. For all faces flagged for deletion, we count the unique indices for each vertex and collapse edges shared by faces with different labels by merging the vertex with the lower index count into the one with the higher index count. Finally, we perform the standard quadric edge collapsing~\citep{10.1145/258734.258849} to dissolve the remaining flagged faces.

We find this routine highly effective for dissolving artifacts unique to foam-based representation. Its effectiveness is currently bounded by our heuristic face flagging, and we expect future learning-based methods to better identify faces for deletion. Note that we don't apply this routine for metrics computation to ensure fairness.

\begin{figure}[h]
    \centering
    \includegraphics[width=1\linewidth]{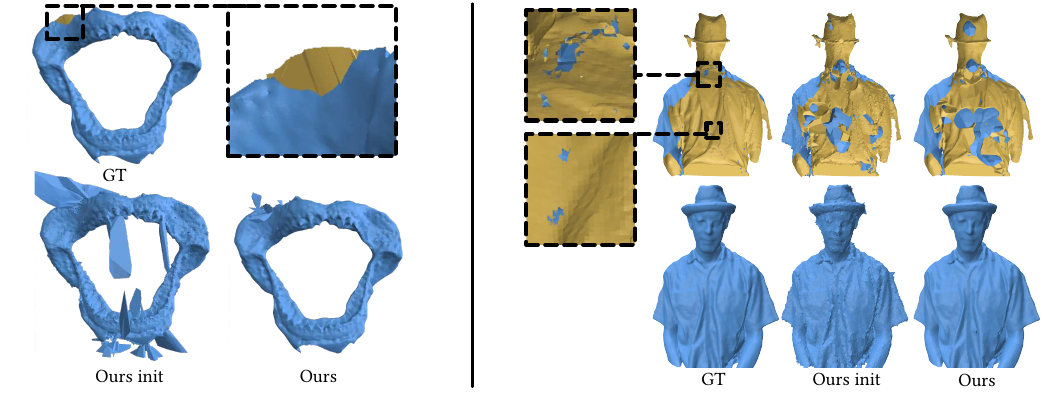}
    \vspace{-10pt}
    \caption{Two failure modes of Thingi10K dataset. Left, the GT surface has self-intersections and negated surface normals (yellow region in the zoom-in view); right, inside the bounded shape there are floating open polygon soups resulting from scanning noise (blue regions in the zoom-in views). Applying our method to both cases can produce spurious parts, degrading the Hausdorff Distance. The visualization is performed using Polyscope~\citep{polyscope}.}
    \label{fig:failure}
\end{figure}

\newcommand{\up}[1]{\textcolor{black!50!black}{\scriptsize($#1$)}}
\newcommand{\dn}[1]{\textcolor{black!70!black}{\scriptsize($#1$)}}
\newcommand{\nc}[1]{{\scriptsize($#1$)}}

\begin{table}[t]
\centering
\caption{Sensitivity analysis of defect subsets in the Thingi10K dataset.}
\label{tab:failure-analysis}
\resizebox{\linewidth}{!}{%
\begin{tabular}{c|cccccc}
\hline
 & Method & CD$\downarrow$ & HD$\downarrow$ & F1$\uparrow$ & NC$\uparrow$ & IoU$\uparrow$ \\
\hline
\multirow{3}{*}{Baseline} & IGR & 1.94 & 4.08 & 58.60 & \textbf{98.95} & 98.06 \\
 & SPSR & 1.88 & \textbf{4.06} & 65.71 & 97.82 & 97.59 \\
 & \textbf{Ours} & \textbf{1.80} & 6.78 & \textbf{77.27} & 98.17 & \textbf{99.14} \\
\hline
\multirow{3}{*}{w.o. Non-manifold} & IGR & 1.93 \up{-0.01} & \textbf{3.87} \up{-0.21} & 58.60 \nc{0.00} & \textbf{98.95} \nc{0.00} & 98.07 \up{+0.01} \\
 & SPSR & 1.91 \dn{+0.03} & 4.08 \dn{+0.02} & 65.12 \dn{-0.59} & 97.76 \dn{-0.06} & 97.48 \dn{-0.11} \\
 & \textbf{Ours} & \textbf{1.74} \up{-0.06} & 5.17 \up{-1.61} & \textbf{76.93} \dn{-0.34} & 98.22 \up{+0.05} & \textbf{99.17} \up{+0.03} \\
\hline
\multirow{3}{*}{w.o. Non-closed} & IGR & 1.92 \up{-0.02} & 3.97 \up{-0.11} & 59.26 \up{+0.66} & \textbf{98.95} \nc{0.00} & 98.04 \dn{-0.02} \\
 & SPSR & 1.87 \up{-0.01} & \textbf{3.93} \up{-0.13} & 66.01 \up{+0.30} & 97.82 \nc{0.00} & 97.55 \dn{-0.04} \\
 & \textbf{Ours} & \textbf{1.51} \up{-0.29} & 5.91 \up{-0.87} & \textbf{77.77} \up{+0.50} & 98.42 \up{+0.25} & \textbf{99.32} \up{+0.18} \\
\hline
\multirow{3}{*}{w.o. Both} & IGR & 1.91 \up{-0.03} & \textbf{3.75} \up{-0.33} & 59.31 \up{+0.71} & \textbf{98.96} \up{+0.01} & 98.05 \dn{-0.01} \\
 & SPSR & 1.90 \dn{+0.02} & 3.94 \up{-0.12} & 65.42 \dn{-0.29} & 97.77 \dn{-0.05} & 97.43 \dn{-0.16} \\
 & \textbf{Ours} & \textbf{1.42} \up{-0.38} & 4.16 \up{-2.62} & \textbf{77.46} \up{+0.19} & 98.50 \up{+0.33} & \textbf{99.37} \up{+0.23} \\
\hline
\end{tabular}
}
\end{table}

\subsection{Failure cases analysis}
For Thingi10K (Tab. \ref{tab:abc_thingi10k}), our method produces a larger Hausdorff Distance (6.78) compared to SPSR (4.06) and IGR (4.08). We analyze the failure cases and find they primarily come from defects in GT geometry. Concretely, there are two failure modes: 1) Non-manifold (ID: 53920, 76277, 77319, 77513, 81528, 92880, 93497, 96637), where the surface has self-intersections and negated surface normals, resulting in overshooting polyhedra (Fig. \ref{fig:failure} left). 2) Non-closed (ID: 64444, 85540, 85546), where inside the bounded shape there are floating open polygon soups resulting from scanning noise (Fig. \ref{fig:failure} right). We perform a sequence of analyses to exclude those subsets for metric evaluation, along with SPSR and IGR, and report results in Tab. \ref{tab:failure-analysis}. With those two sets removed, our Hausdorff distance (HS) drops to a similar level as SPSR and IGR, while our Chamfer distance (CD), F-Score (FS), and Intersection-over-Union (IoU) remain the top throughout the analysis. This suggests a global SDF is more robust to topological defects, while our method excels at accuracy.

\begin{figure}[h]
    \centering
    \includegraphics[width=1\linewidth]{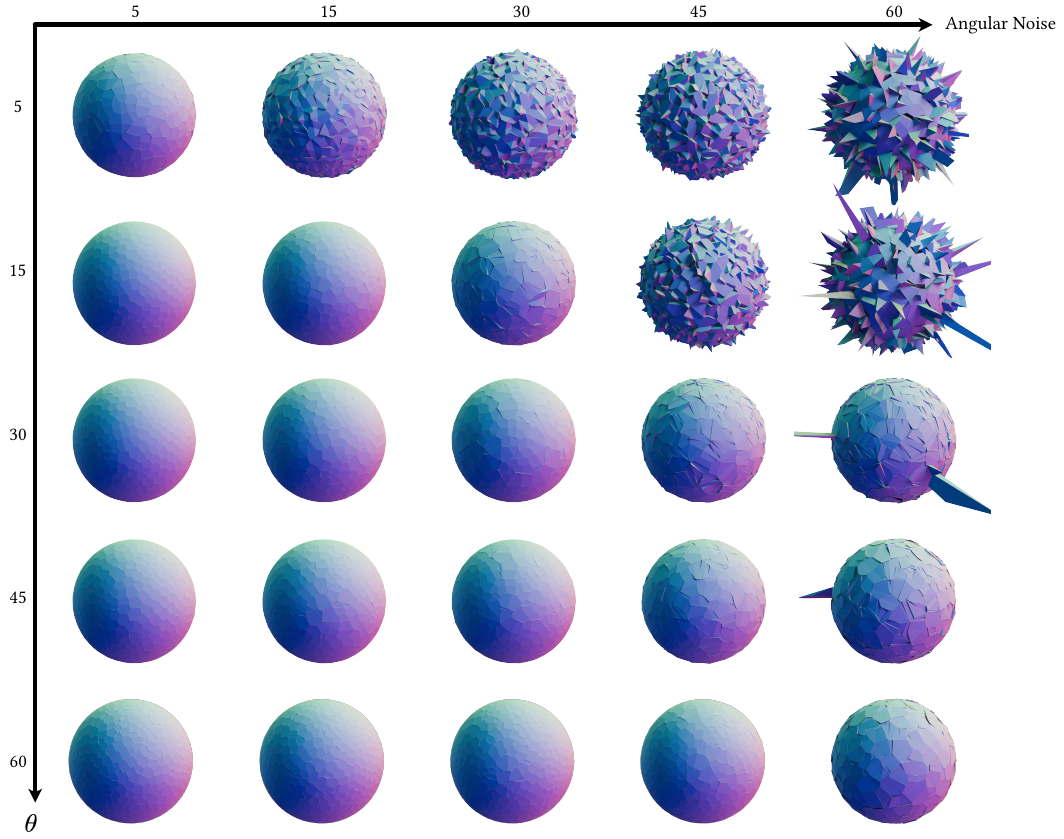}
    \vspace{-10pt}
    \caption{The angular noise tolerance of our method. Each column shows increasing angular noise standard deviation (in degrees), and each row shows increasing angular threshold $\theta$.}
    \label{fig:angular-noise}
\end{figure}

\subsection{Angular Noise Tolerance}
\label{sec:noise-tolerance}

The angular stabilization threshold $\theta$ in Sec. \ref{sec:mainmethod} offers precise control of angular noise tolerance. Normal deviations within this threshold are treated as noise and absorbed into the cell pair, while deviations exceeding it are treated as mis-coverage, and the gradient is reassigned to the most suitable pair in the neighborhood via discrete matching (Fig. \ref{fig:angular-noise}). For clean shapes, we assume a $5\deg$ to account for the error inherent in approximating curved surfaces with flat patches.

\begin{figure}[h]
    \centering
    \includegraphics[width=1\linewidth]{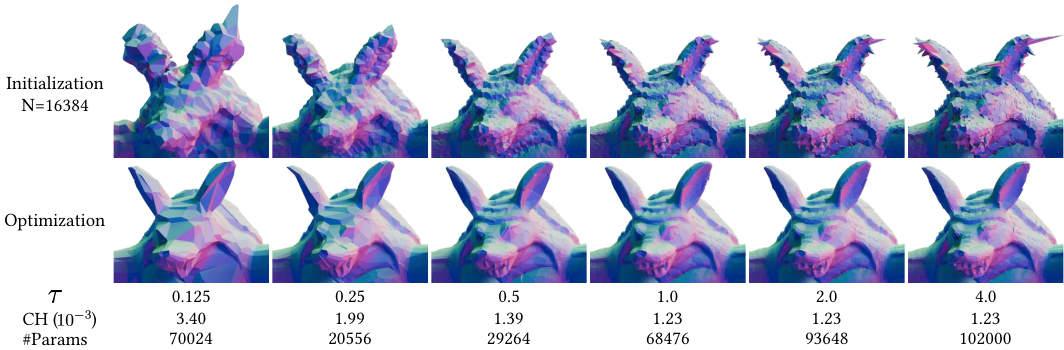}
    \vspace{-10pt}
    \caption{The ablation of hyperparameter $\tau$. All shapes are initialized with $N=16384$ (i.e. 131072 float parameters). The "$\#$Params" refers to the number of parameters after pruning.}
    \label{fig:tau}
\end{figure}

\subsection{Scaling hyperparameter Ablation}
\label{sec:tau-ablation}

The scaling hyperparameter $\tau$ determines the offset $\delta = \frac{0.01}{\tau}$ of Voronoi sites at initialization. As illustrated in Fig. \ref{fig:tau}, if $\delta$ is too large (equivalently, $\tau$ is too small) and overshoots the primary features, the field becomes noisy, the candidate set fails to converge for pruning, and the resulting geometry is noisy. If $\delta$ overshoots only small features while preserving the large structure, fine geometric details are lost and cannot be recovered. Decreasing $\delta$ (equivalently, increasing $\tau$) recovers more detail but increases the final parameters, and beyond a certain point the quality gains diminish while the parameter count keeps growing. In practice, we choose $\tau$ by visualizing our initialization, which is training-free and can be meshed directly via Marching Cubes or explicit polygonization. As demonstrated in our experiments, a single $\tau$ is reasonably robust across all shapes within a dataset, despite the varying feature sizes within each shape.

\subsection{Pseudocode}
\label{sec:pseudo-code}

The pseudocode of our main algorithm is available in Fig. \ref{fig:code}.

\begin{figure}[h]
\begin{lstlisting}[language=Python]
def l(a, c, idx, x):
    return dot(a[idx], x) + c[idx]

x_init, n_init = mesh.sample(N)

# Geometric Initialization
a_pos = 2 * (x_init + tau * n_init);   c_pos = -0.25 * |a_pos|^2
a_neg = 2 * (x_init - tau * n_init);   c_neg = -0.25 * |a_neg|^2

# Auxiliary score
s_pos, s_neg = ones(N), ones(N)

optim = Adam([a_pos, c_pos, a_neg, c_neg, s_pos, s_neg], lr=3e-4 / tau)

for iter in range(10000):
    x_k, n_k = mesh.sample(K)

    idx_max_pos, idx_sec_max_pos = maximum_query(x_k, a_pos, c_pos,
      mask=s_pos > 0.5)
    idx_max_neg, idx_sec_max_neg = maximum_query(x_k, a_neg, c_neg,
      mask=s_neg > 0.5)

    # Discrete Matching
    pairs = [(p, q) for p in (idx_max_pos, idx_sec_max_pos)
                    for q in (idx_max_neg, idx_sec_max_neg)]
    candidate_n = stack([a_pos[p] - a_neg[q] for p, q in pairs])
    dp = cosine_similarity(candidate_n, n_k)
    pick = argmax(dp)

    # Assignment Stabilization
    d_theta = arccos(dp[0]) - arccos(dp[pick])
    pick = pick if |d_theta| > 5° else 0
    idx_pick_pos, idx_pick_neg = pairs[pick]

    # Cell Fitting
    f = l(a_pos, c_pos, idx_pick_pos, x_k) 
      - l(a_neg, c_neg, idx_pick_neg, x_k)
    cos = cosine_similarity(a_pos[idx_pick_pos] 
                          - a_neg[idx_pick_neg], n_k)
    loss_sur    = mean(|f|)
    loss_normal = mean(1 - cos)
    loss_cover  = mean(relu(l(a_pos, c_pos, idx_max_pos, x_k) 
                          - l(a_pos, c_pos, idx_pick_pos, x_k))
                     + relu(l(a_neg, c_neg, idx_max_neg, x_k) 
                          - l(a_neg, c_neg, idx_pick_neg, x_k)))

    # Pruning
    s_mask_pos = not (idx_max_pos or idx_sec_max_pos) or idx_pick_pos
    # Normal Inversion Handling
    s_mask_neg = not (idx_max_neg or idx_sec_max_neg) or idx_pick_neg[cos > 0]
    loss_prune = mean(|s_pos - s_mask_pos|) + mean(|s_neg - s_mask_neg|)

    loss = 10 * loss_sur + loss_normal + 100 * loss_cover + loss_prune
    optim.step(loss)

keep_pos, keep_neg = s_pos > 0.5, s_neg > 0.5
save_patchwork(a_pos[keep_pos], c_pos[keep_pos],
               a_neg[keep_neg], c_neg[keep_neg])
\end{lstlisting}
\caption{Main algorithm pseudocode.}
\label{fig:code}
\end{figure}

\section{Implicit Fitting Formulation}
\label{sec:implicit-fitting}
In this section, we describe one variant of the approach that fits patchwork directly as a standard implicit function. This is the variant we tune as best we can, and we use it for an ablation study of our discrete matching scheme.

\subsection{Formation}
We consider patchwork in the following augmented form:
\begin{equation}
\label{eq-augmented-form}
\begin{aligned}
F(x)&=\frac{1}{\beta^{+}}
\log\sum_{i=1}^{n^+} \exp\left(\beta^+(\langle a^{+}_{i}, x\rangle +c^{+}_{i}) + \log{s_i^+}\right) \\
&- \frac{1}{\beta^{-}}
\log\sum_{i=1}^{n^-} \exp\left(\beta^{-}(\langle a^{-}_{i}, x\rangle+c^{-}_{i}) + \log{s_i^-}\right),
\end{aligned}
\end{equation}
where $n^+ + n^- = n$, $a_i^{\pm} \in \mathbb{R}^3$, $c_i^{\pm} \in \mathbb{R}$, $s_i^{\pm}, \beta^{\pm} \in \mathbb{R}_{>0}$. Note that $s_i^{\pm}$ can be losslessly incorporated into $c_i^{\pm}$, and we keep them here only for pruning purposes.

Given an oriented point cloud $\{x_j, n_j\}_{j=1}^m$, where $x_j\in \mathbb{R}^3$ are points and $n_j\in \mathbb{R}^3$ are normal vectors. We aim to fit $F(x)$ such that its zero level set $\Omega_0 =\{x\in\mathbb{R}^3|F(x)= 0\}$ describes the underlying geometry. To this end, we minimize the surface-approximation error and the normal-misalignment penalty:
\begin{equation}
\begin{aligned}
\mathcal{L}_{\text{sur}} &= \frac{1}{m} \sum_{j=1}^m |F(x_j)|
\quad\text{and}\quad
\mathcal{L}_{\text{normal}} = \frac{1}{m} \sum_{j=1}^m \left(1 - \frac{\nabla F(x_j) \cdot n_j}{\|\nabla F(x_j)\|\|n_j\|}\right).
\end{aligned}
\end{equation}

To reduce off-surface parts of $\Omega_0$, we additionally sample $m$ off-surface samples $\{y_k\}_{k=1}^m, y_k \in \mathbb{R}^3 \setminus \Omega_0$ per iteration for a soft occupancy regularization
\begin{equation}
\begin{aligned}
g_{dw}(x) &= 4 (x - 0.5)^2 - 4 |x - 0.5| + 1, \\
\mathcal{L}_{\text{reg}} &= \frac{1}{m} \sum_{k=1}^{m} g_{dw}(\text{sigmoid}(-F(y_k)))^2,
\end{aligned}
\end{equation}
where $g_{dw}$ is the double-well potential \citep{lipman2021phasetransitionsdistancefunctions} that attains minimum 
at $x=0, 1$.

The regularization $\mathcal{L}_{\text{reg}}$ additionally forces some candidate polyhedra to expand, leaving the other ones redundant. We hence introduce $\mathcal{L}_{\text{prune}}$ in Sec.~\ref{subsec:optimization} to progressively eliminate unnecessary ones during the fitting process.

Our total losses are
$$
\mathcal{L}_{\text{total}} = \mathcal{L}_{\text{sur}} + \mathcal{L}_{\text{normal}} + \mathcal{L}_{\text{reg}} + \mathcal{L}_{\text{prune}}.
$$

\subsection{Explicit geometric initialization}
From an optimization point of view, the Patchwork falls into the broad category of Difference-of-Convex (DC) problems. Specifically, its limiting form is the difference of two convex functions. Inspired by Piecewise Linear Regression (PLR) \citep{siahkamari2020piecewise} in DC literature, we assign each input point $x_j$ two linear functions (i.e., set $n^+ = n^- = m$) and force the two maxima at $x_j$ to be attained with their corresponding coefficients $\{a_j^+, c_j^+, a_j^-, c_j^-\}$. Formally,
\begin{equation} \label{eq-PLR-assumptions}
\begin{aligned}
&f(x) = \max_{i} \{\langle a_i^+, x \rangle + c_i^+\} - \max_{i} \{\langle a_i^-, x \rangle + c_i^-\}, \\
& \arg \max_{i} \{\langle a_i^+, x_j\rangle + c_i^+\} = \arg \max_{i} \{\langle a_i^-, x_j\rangle + c_i^-\}=j.
\end{aligned}
\end{equation}
The two pointwise global-maximum constraints significantly reduce the solution space, yet the formulation remains a DC, a universal approximator. 

To further reduce the solution space, we assume $f(x)$ at the initialization approximates the signed distance function. Inspired by proximal DC decomposition \citep{sun2024understandeffectivenessshortcutslens}, we perform the following construction
\begin{equation}\label{eq-hpm}
h^+(x) = \tfrac{1}{2} \rho \|x\|^2 + \text{sdf}(x), \ h^-(x) = \tfrac{1}{2} \rho \|x\|^2, \ f(x) \approx h^+(x) - h^-(x),
\end{equation}
where $\rho \in \mathbb{R}$ is some 
parameter large enough to make $h^+(x)$ convex.

Conviently, with oriented surface samples $\{x_j, n_j\}$, we can set $f(x_j)=\text{sdf}(x_j) = 0$ and $\nabla f(x_j)=\nabla \text{sdf}(x_j) = n_j$, yeilding the closed-form solution
\begin{equation}\label{eq-geometric-initialization}
\begin{aligned}
a^+_j &= \nabla h^+(x_j) = \rho x_j + n_j, &
c_j^+ &= h^+(x_j) - \langle a^+_j, x_j \rangle = -\tfrac{1}{2} \rho \|x_j\|^2 -\langle n_j, x_j \rangle, \\
a^-_j &= \nabla h^-(x_j) = \rho x_j, &
c_j^- &= h^-(x_j) - \langle a^-_j, x_j \rangle = -\tfrac{1}{2} \rho \|x_j\|^2.
\end{aligned}
\end{equation}
We defer the detailed derivation to \journalvsarxiv{the supplementary}{Appendix~\ref{sec:geometric-initialization-derivation}}. In practice, we initialize with $\rho=200$, $s_i^{\pm}=1$, $\beta^\pm=75$ across all our experiments. Notice the norm of $a_j^\pm$ scales with $\rho$, making updating its direction stiff. So we apply weight normalization to disentangle its magnitude and direction gradient.

\subsection{Progressive pruning}
\label{subsec:optimization}
We propose the following loss to progressively eliminate non-contributing linear functions
\begin{equation}
\label{eq-loss-prune}
\begin{aligned}
w^{\pm}_i(x) &= \text{sg}\left(\frac{\beta^{\pm}(\langle a^{\pm}_{i}, x\rangle +c^{\pm}_{i}) + \log s^{\pm}_i}{\sum_{l=1}^{n^{\pm}} \left(\beta^\pm(\langle a^{\pm}_{l}, x\rangle +c^{\pm}_{l}) + \log s^{\pm}_l\right)}\right), \\
\mathcal{L}_{\text{prune}}^\pm &= \frac{1}{n^{\pm}} \sum_{i=1}^{n^{\pm}} s_i^{\pm} + \frac{1}{m} \sum_{j=1}^m \text{ReLU}\left(1 - \sum_{i=1}^{n^{\pm}} s_i^\pm \cdot w^{\pm}_i(x_j)\right), \\
\mathcal{L}_{\text{prune}} &= \mathcal{L}_{\text{prune}}^+ + \mathcal{L}_{\text{prune}}^-,
\end{aligned}
\end{equation}
where $\text{sg}$ is the stop gradient operator, and $w^{\pm}_i$ can be viewed as a  typical soft maximum. The first term encourages all $s^{\pm}_i$ to be diminished, while the second term maintains the magnitude of $s^{\pm}_i$ that contributes to candidate polygons.

For every fixed iterations, we detect linear functions with $s^{\pm}_i < 10^{-2}$, and disable them softly by setting $s^{\pm}_i = 10^{-5}$.

\subsection{Initialization derivation}
\label{sec:geometric-initialization-derivation}

Let us derive geometric initialization~(\ref{eq-geometric-initialization}) from assumptions~(\ref{eq-PLR-assumptions}).
We start with introducing the functions 
\begin{equation} \label{eq-fpm-def}
f^\pm(x) = \max_{i} \{\langle a_i^\pm, x \rangle + c_i^\pm\}.
\end{equation}
By our assumptions~(\ref{eq-PLR-assumptions}), the maximum for $x=x_j$ is attained at $i=j$. Then it is natural to assume that for $x$ sufficiently close to $x_j$, 
we still have
\begin{equation}\label{eq-fpmx}
f^\pm(x) = \langle a_j^\pm, x \rangle + c_j^\pm
\end{equation}
and hence
\begin{equation}
\nabla f^\pm(x) = a_j^\pm.
\end{equation}
Substituting $x=x_j$, we get
\begin{equation}\label{eq-fpm}
f^\pm(x_j) = \langle a_j^\pm, x_j \rangle + c_j^\pm
\qquad\text{and}\qquad \nabla f^\pm(x_j) = a_j^\pm.
\end{equation}
Recall definition~(\ref{eq-hpm}) of convex functions 
\begin{equation}
h^+(x) = \tfrac{1}{2} \rho \|x\|^2 + \text{sdf}(x)\qquad\text{and}\qquad h^-(x) = \tfrac{1}{2} \rho \|x\|^2.
\end{equation}
To satisfy 
\begin{equation}
f(x_j)=\text{sdf}(x_j) = 0\qquad\text{and}\qquad\nabla f(x_j)=\nabla \text{sdf}(x_j) = n_j,
\end{equation}
it is natural to require the conditions
\begin{equation}
f^\pm(x_j)=h^\pm(x_j)\qquad\text{and}\qquad\nabla f^\pm(x_j)=\nabla h^\pm(x_j).
\end{equation}
Comparing with equation~(\ref{eq-fpm}), we get the desired closed-form solution~(\ref{eq-geometric-initialization}):
\begin{equation}
\begin{aligned}
a^+_j &= \nabla h^+(x_j) = \rho x_j + n_j, &
c_j^+ &= h^+(x_j) - \langle a^+_j, x_j \rangle = -\tfrac{1}{2} \rho \|x_j\|^2 -\langle n_j, x_j \rangle, \\
a^-_j &= \nabla h^-(x_j) = \rho x_j, &
c_j^- &= h^-(x_j) - \langle a^-_j, x_j \rangle = -\tfrac{1}{2} \rho \|x_j\|^2.
\end{aligned}
\end{equation}
The resulting solution automatically satisfies assumptions~(\ref{eq-PLR-assumptions}): The plot of linear function~(\ref{eq-fpmx}) is the tangent plane to the plot of the convex function $h^\pm(x)$ at the point $x=x_j$. Thus, the maximum in~(\ref{eq-fpm-def}) for $x=x_j$ is attained at $i=j$.

%% file: fig/hex/hex.tex
\begin{figure}[hbt]
    \centering
    \includegraphics[height=2.5cm]{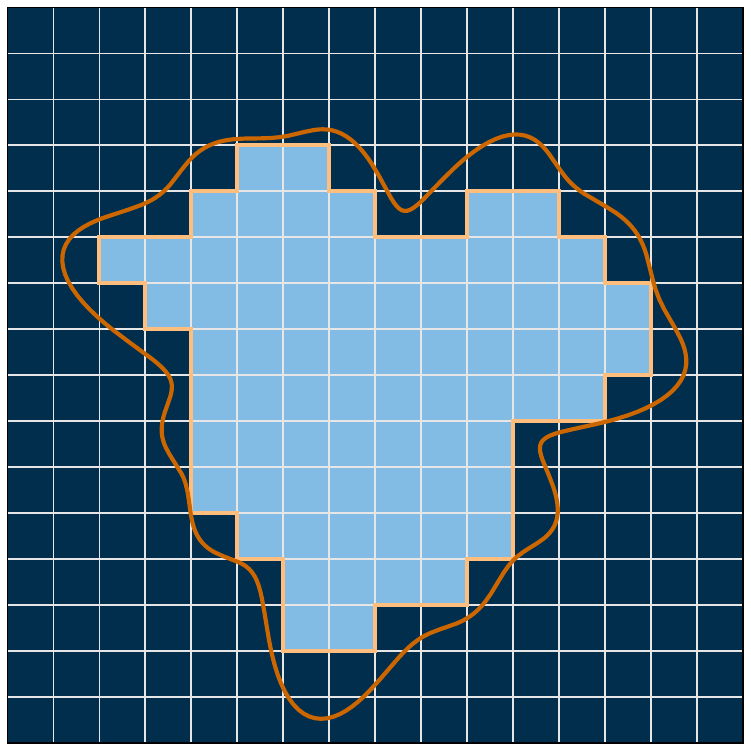}
    \includegraphics[height=2.5cm]{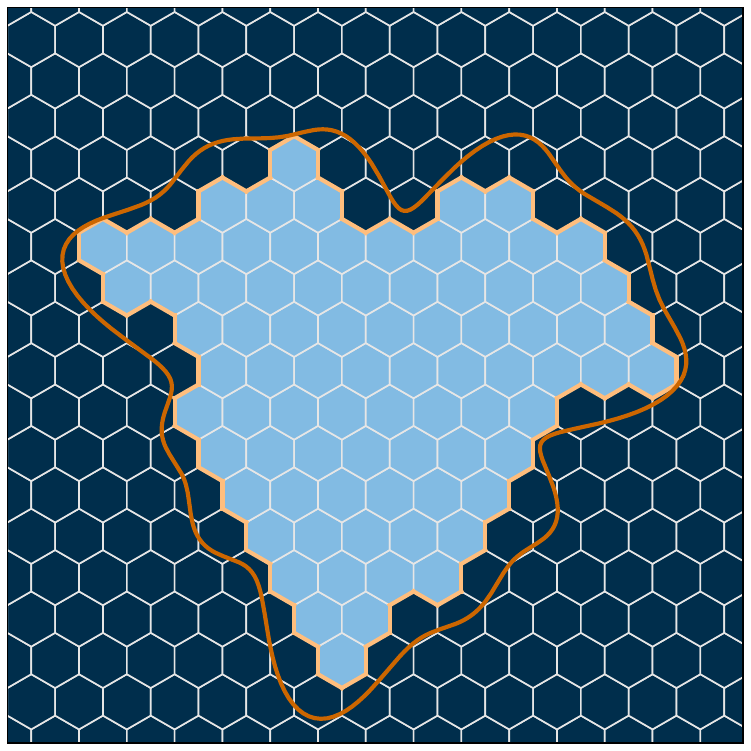}
    \includegraphics[height=2.5cm]{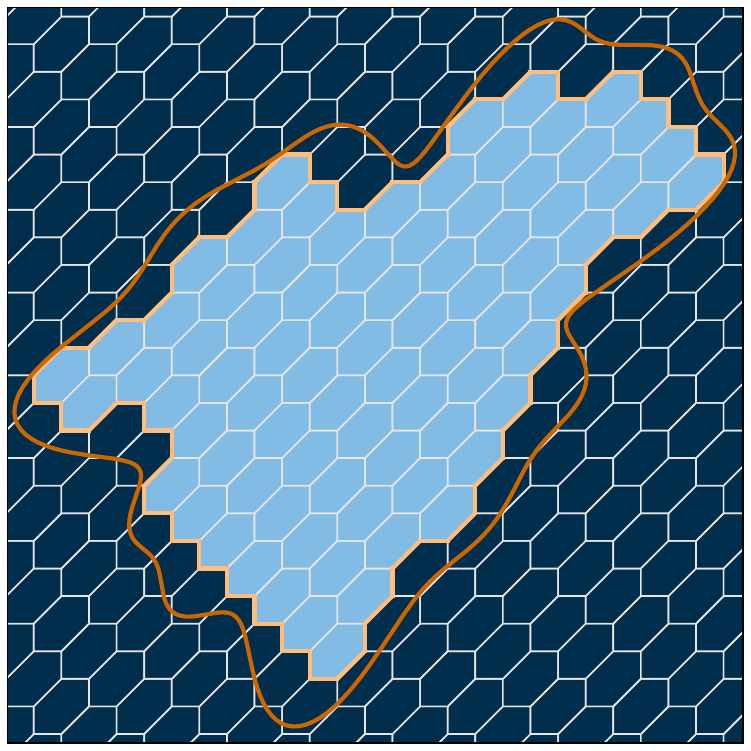}
    \caption{Patchworks (orange) approximating a given curve (red). The construction can employ a square grid (left), a honeycomb lattice (middle), or an affine transformation of the latter (right). The squares/hexagons inside the curve are light, the other ones are dark, and the patchwork separates the colors.}
    \label{fig:hex}
\end{figure}